\documentclass[sigconf,screen]{acmart}

\pdfoutput=1

\AtBeginDocument{%
  \providecommand\BibTeX{{%
    \normalfont B\kern-0.5em{\scshape i\kern-0.25em b}\kern-0.8em\TeX}}}

\setcopyright{acmlicensed}
\acmPrice{15.00}
\acmDOI{10.1145/3468264.3473928}
\acmYear{2021}
\copyrightyear{2021}
\acmSubmissionID{fse21ind-p70-p}
\acmISBN{978-1-4503-8562-6/21/08}
\acmConference[ESEC/FSE '21]{Proceedings of the 29th ACM Joint European Software Engineering Conference and Symposium on the Foundations of Software Engineering}{August 23--28, 2021}{Athens, Greece}
\acmBooktitle{Proceedings of the 29th ACM Joint European Software Engineering Conference and Symposium on the Foundations of Software Engineering (ESEC/FSE '21), August 23--28, 2021, Athens, Greece}

\usepackage[utf8]{inputenc}
\usepackage[T1]{fontenc}
\usepackage{microtype}

\usepackage[nolist]{acronym}
\usepackage[ruled,vlined,linesnumbered]{algorithm2e}
\usepackage{multirow}
\usepackage{subfigure}
\usepackage{balance}
\usepackage[inline]{enumitem}

\newcommand{\TargetVPL}{\textsc{OutSystems}\xspace}

\newcommand{\ActFlowStart}{\textsc{Start}\xspace}
\newcommand{\ActFlowEnd}{\textsc{End}\xspace}
\newcommand{\ActFlowInstruction}{\textsc{Instruction}\xspace}
\newcommand{\ActFlowLoop}{\textsc{ForEach}\xspace}
\newcommand{\ActFlowIf}{\textsc{If}\xspace}
\newcommand{\ActFlowSwitch}{\textsc{Switch}\xspace}

\newcommand{\ActFlowConnector}{\textsc{Connector}\xspace}
\newcommand{\ActFlowTrue}{\textsc{True}\xspace}
\newcommand{\ActFlowFalse}{\textsc{False}\xspace}
\newcommand{\ActFlowCycle}{\textsc{Cycle}\xspace}
\newcommand{\ActFlowCondition}{\textsc{Condition}\xspace}
\newcommand{\ActFlowOtherwise}{\textsc{Otherwise}\xspace}

\newcommand{\Greedy}{\textsc{Greedy}\xspace}
\newcommand{\Lazy}{\textsc{Lazy}\xspace}
\newcommand{\DedupThenLazy}{\textsc{DedupThenLazy}\xspace}
\newcommand{\DedupThenLazyInvIndex}{\textsc{DedupThenLazy+Index}\xspace}

\newcommand{\SourcererCC}{\textsc{SourcererCC}\xspace}
\newcommand{\NiCaD}{\textsc{NiCaD}\xspace}

\newcommand{\ConQAT}{\textsc{ConQAT}\xspace}
\newcommand{\eScan}{\textsc{eScan}\xspace}
\newcommand{\ScanQAT}{\textsc{ScanQAT}\xspace}
\newcommand{\SIMONE}{\textsc{SIMONE}\xspace}

\newcommand{\CCGraph}{\textsc{CCGraph}\xspace}
\newcommand{\CCSharp}{\textsc{CCSharp}\xspace}

\newcounter{inlineenum}
\renewcommand{\theinlineenum}{\alph{inlineenum}}
\newenvironment{inlineenum}
  {\unskip\ignorespaces\setcounter{inlineenum}{0}%
   \renewcommand{\item}{\refstepcounter{inlineenum}{\textit{\theinlineenum})~}}}
  {\ignorespacesafterend}

\begin{document}

\title[Duplicated Code Pattern Mining in Visual Programming Languages]{Duplicated Code Pattern Mining in \\ Visual Programming Languages}
%\title{Duplicated Code Detection in Low-Code Development Platforms}
\titlenote{This is an extended version of a paper accepted for publication in the industrial track of the Symposium on the Foundations of Software Engineering (FSE) 2021.}

\author{Miguel Terra-Neves}
\affiliation{%
  \institution{OutSystems}
  \country{Portugal}
}
\email{miguel.neves@outsystems.com}

\author{João Nadkarni}
\affiliation{%
  \institution{OutSystems}
  \country{Portugal}
}
\email{joao.nadkarni@outsystems.com}

\author{Miguel Ventura}
\affiliation{%
  \institution{OutSystems}
  \country{Portugal}
}
\email{miguel.ventura@outsystems.com}

\author{Pedro Resende}
\affiliation{%
  \institution{OutSystems}
  \country{Portugal}
}
\email{pedro.resende@outsystems.com}

\author{Hugo Veiga}
\affiliation{%
  \institution{OutSystems}
  \country{Portugal}
}
\email{hugo.veiga@outsystems.com}

\author{António Alegria}
\affiliation{%
  \institution{OutSystems}
  \country{Portugal}
}
\email{antonio.alegria@outsystems.com}

%\renewcommand{\shortauthors}{Terra-Neves, et al.}
%\renewcommand{\shortauthors}{Miguel Terra-Neves, João Nadkarni, Miguel Ventura, Pedro Resende, Hugo Veiga, and António Alegria}

\begin{acronym}
    \acro{VPL}{Visual Programming Language}
    \acro{LCDP}{Low-Code Development Platform}
    \acro{CNF}{Conjunctive Normal Form}
    \acro{SAT}{Boolean Satisfiability}
    \acro{MaxSAT}{Maximum Satisfiability}
    \acro{MCS}{Maximum Common Sub-graph}
    \acro{AST}{Abstract Syntax Tree}
    \acro{PDG}{Program Dependence Graph}
\end{acronym}

\begin{abstract}
    \acp{VPL}, coupled with the high-level abstractions that are commonplace in visual programming environments, enable users with less technical knowledge to become proficient programmers.
    However, the lower skill floor required by \acp{VPL} also entails that programmers are more likely to not adhere to best practices of software development, producing systems with high technical debt, and thus poor maintainability.
    Duplicated code is one important example of such technical debt.
    In fact, we observed that the amount of duplication in the \TargetVPL \ac{VPL} code bases can reach as high as $39\%$.
    
    Duplicated code detection in text-based programming languages is still an active area of research with important implications regarding software maintainability and evolution.
    However, to the best of our knowledge, the literature on duplicated code detection for \acp{VPL} is very limited.
    We propose a novel and scalable duplicated code pattern mining algorithm that leverages the visual structure of \acp{VPL} in order to not only detect duplicated code, but also highlight duplicated code patterns that explain the reported duplication.
    The performance of the proposed approach is evaluated on a wide range of real-world mobile and web applications developed using \TargetVPL.
\end{abstract}




\maketitle

\acresetall

\section{Introduction}
\label{sec:intro}

\acp{VPL} allow users to describe computational processes in terms that are easier for humans to understand than text-based programming languages.
Additionally, some \acp{VPL} provide high-level abstractions that simplify and speed-up the development process, as is the case of \TargetVPL\footnote{https://www.outsystems.com/}.
This results in a low entry barrier that enables users with less technical background to become proficient programmers.
However, such users are more likely to write code with high technical debt, since these are less familiar with best practices of software development.

In this work, we aim to aid \TargetVPL developers manage one important form of technical debt: duplicated code.
Duplicated code is commonplace in software developed using traditional text-based languages~\cite{analysis-clones,text-lang-independent} and may have severe adverse effects that result in higher maintenance costs.
For example, if one changes a duplicated code block, it is likely that the same change may need to be applied to most, if not all, duplicates of that block, thus making software harder to evolve and maintain.
Code duplication may also exacerbate bug propagation, since a bug in a given code block will also be present in its copies.
In our experiments, we observed that the amount of duplicated code in real-world \TargetVPL code bases can reach as high as $39\%$, highlighting the importance of addressing code duplication in \TargetVPL.

\begin{figure}[t]
    \centering
    \includegraphics[width=0.47\textwidth]{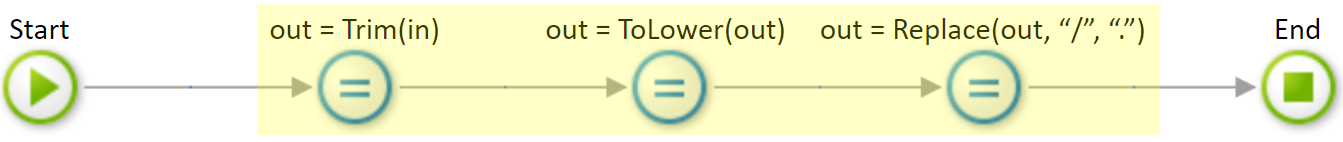}
    \caption{A logic flow that transforms a single string.}
    \label{fig:act-flow-simple}
\end{figure}

\begin{figure}[t]
    \centering
    \includegraphics[width=0.43\textwidth]{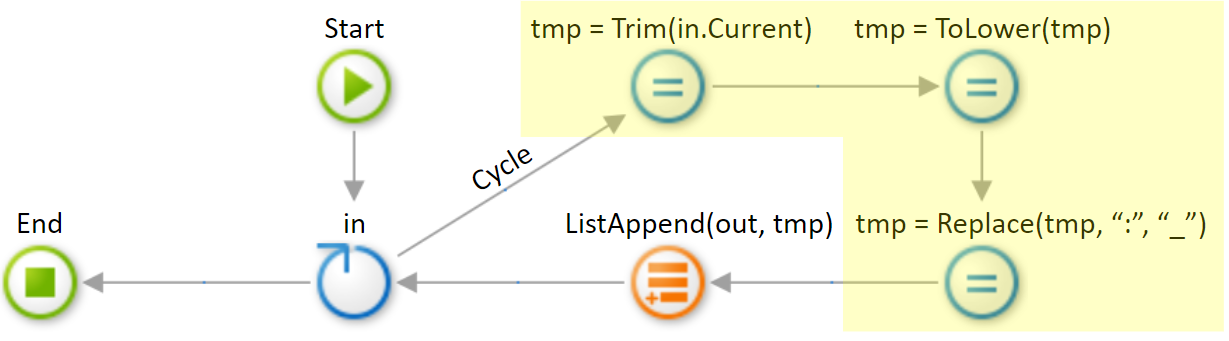}
    \caption{A logic flow that transforms a list of strings.}
    \label{fig:act-flow-foreach}
\end{figure}

In \TargetVPL, logic is implemented through logic flows.
Figure~\ref{fig:act-flow-simple} shows an example of a flow that performs some transformations over some input string.
The goal is to detect a limited form of type 3 duplicates~\cite{clone-types}, where near-misses are allowed for node expressions but the graph structure of the duplicated part must be the same.
Moreover, this duplicated structure must be visually highlighted to the user, and thus the duplicated code detector must return the mappings of flow nodes to the duplicated code pattern nodes.
These requirements stem from discussions with \TargetVPL experts, regarding an earlier version of our tool, that exploited data dependencies between nodes in order to find duplicated code with significant syntactic differences.
We concluded that such duplicates were hard to analyse and understand, thus negatively impacting the user experience.
Figure~\ref{fig:act-flow-foreach} shows an example of a flow that performs the same transformations as in Figure~\ref{fig:act-flow-simple} over some list of strings.
The respective duplicated code pattern is highlighted in yellow.
In addition to the aforementioned functional requirements, the duplicated code detector must be integrated in a tool that performs static analyses for hundreds of \TargetVPL code bases every 12 hours.
Nonetheless, the detector should process these code bases as fast as possible in order to minimize the computational resources needed to satisfy this time limit, thus optimizing operating costs.

A naive approach for detecting duplicated code in software developed using a \ac{VPL} could be to translate from the \ac{VPL} to some text-based language and then apply one of many detectors for such languages~\cite{nicad-plus,sourcerercc,siamese,deckard,ccgraph,scdetector}.
This approach suffers from a severe drawback: it sacrifices the visual structure of the \ac{VPL} code, which can be leveraged in order to provide helpful explanations of reported duplications by highlighting duplicated code patterns.
Such patterns allow the developer to understand and address the sources of code duplication more effectively.
Alternatively, some graph-based algorithms for text-based languages~\cite{graph-pdg-similar,gplag,ccsharp,ccgraph} or other \acp{VPL}~\cite{conqat,escan,scanqat,simone,opmcd} could be directly applied to \TargetVPL logic, but these typically suffer from scalability issues due to the hardness of checking sub-graph isomorphism, and the ones that do address this issue perform some approximated form of sub-graph matching, thus not guaranteeing the consistency of the graph structure.

We propose a duplicated code detector for \TargetVPL that addresses the aforementioned issues by iteratively mining \acp{MCS} of graph representations of \TargetVPL code.
Our main contributions are as follows:
\begin{inlineenum}
    \item Several complete graph pre-processing techniques that simplify the \ac{MCS} extraction task.
    We use these techniques to improve the efficiency of an \ac{MCS} algorithm based on \ac{MaxSAT}.
    \item A novel and scalable greedy algorithm for mining duplicated code patterns in \TargetVPL code bases.
    Although the focus of this work is on duplicated code, the proposed algorithm is generic and can thus be used to mine \acp{MCS} of arbitrary graph structures.
    Some techniques are also proposed in order to improve the performance of the mining algorithm.
    To the best of our knowledge, ours is the first graph-based approach that solves the scalability issue by using an inverted index~\cite{sourcerercc}.
    \item An extensive experimental evaluation on real-world \TargetVPL code bases that assess the performance of the proposed techniques.
    \item A brief evaluation regarding the severity of code duplication in real-world \TargetVPL code bases.
\end{inlineenum}

We start by providing some background on \TargetVPL, \acp{MCS} and \ac{MaxSAT} in Section~\ref{sec:bg}.
Then, the \ac{MaxSAT}-based \ac{MCS} algorithm and graph pre-processing techniques are explained in Section~\ref{sec:spe}, followed by the pattern mining algorithm and respective performance improvements in Section~\ref{sec:mining}.
Experimental results showing the merits of the proposed techniques are presented in Section~\ref{sec:eval}.
Section~\ref{sec:related-work} summarizes related work on duplicated code detection and sub-graph mining.
Limitations and design decisions are discussed in Section~\ref{sec:limits-discuss}.
Finally, Section~\ref{sec:conclusion} concludes this paper.

\section{Background}
\label{sec:bg}

In this section, we introduce the necessary background.
Logic flows are explained in Section~\ref{sec:bg-act-flows}, followed by a definition of \ac{MCS} in Section~\ref{sec:bg-mcs} and an explanation of \ac{MaxSAT} in Section~\ref{sec:bg-maxsat}.

\subsection{Logic Flows}
\label{sec:bg-act-flows}

A logic flow is a directed weakly connected graph $G = (V, E)$ where each node in $V$ has one of the following types: \ActFlowStart, \ActFlowEnd, \ActFlowInstruction, \ActFlowLoop, \ActFlowIf~or \ActFlowSwitch.
Additionally, each edge in $E$ can be of type \ActFlowConnector, \ActFlowTrue, \ActFlowFalse, \ActFlowCycle, \ActFlowCondition~or \ActFlowOtherwise.
We refer to the outgoing edges of a node as branches.
$G$ satisfies the following properties:
\begin{itemize}
    \item $G$ does not contain self-loops or parallel edges.
    \item $V$ contains only one \ActFlowStart~node $v$, and no edge $(u', v') \in E$ exists such that $v = v'$.
    \item Given an \ActFlowEnd~node $v \in V$, no branch exists in $E$ for $v$ and there exists at least one edge $(u', v') \in E$ such that $v = v'$.
    \item A \ActFlowStart~or \ActFlowInstruction~node $u \in V$ has exactly one \ActFlowConnector~branch $(u, v) \in E$.
    \item An \ActFlowIf~node $u \in V$ has exactly one \ActFlowTrue~branch $(u, v) \in E$ and one \ActFlowFalse~branch $(u, v') \in E$.
    \item A \ActFlowLoop~node $u \in V$ has exactly one \ActFlowConnector~branch $(u, v) \in E$ and one \ActFlowCycle~branch $(u, v') \in E$ such that there exists a path from $u$ to itself through $(u, v')$.
    \item A \ActFlowSwitch~node $u \in V$ has at least one \ActFlowCondition~branch $(u, v) \in E$ and exactly one \ActFlowOtherwise~branch $(u, v') \in E$.
\end{itemize}

The logic flow is akin to the control flow graph of a program written in a traditional programming language.
Its execution begins at its \ActFlowStart~node and terminates at one of its \ActFlowEnd~nodes.
Moreover, depending on their types, the nodes/edges can have different attributes.
For example, an \ActFlowIf node contains a Boolean expression which dictates if the execution is to continue through its \ActFlowTrue~(\ActFlowCycle) or \ActFlowFalse~(\ActFlowConnector) branch.
Similarly, a \ActFlowCondition~branch of a \ActFlowSwitch~node contains a Boolean expression that, if evaluated to true, then the execution continues through that branch.
\ActFlowCondition~branches also have a pre-specified order of evaluation.
If none of those branches evaluate to true, then execution resumes through the \ActFlowOtherwise~branch.
A \ActFlowLoop~node contains a reference to a variable of an iterable type (e.g. list).
\ActFlowInstruction~nodes can be of various kinds, such as variable assignments, database accesses, calls to other logic flows, among others.
Note that, just like functions/methods in text-based languages, logic flows can have input and output parameters.

\subsection{Maximum Common Sub-graph}
\label{sec:bg-mcs}

%Let $G_1 = (V_1, E_1)$ and $G_2 = (V_2, E_2)$ be a pair of graphs with labeled nodes/edges.
%For the purpose of this work, we assume that graphs are directed by default.
%We use $L(v)$ to denote the label of some node $v$.
%For example, assuming $v$ is a node of a logic flow, $L(v)$ can be something as simple as the node's type (e.g. \ActFlowInstruction).
Logic flows are graphs, thus a duplicated code pattern is a common sub-graph that occurs across multiple flows.
Naturally, the largest common pattern in those flows corresponds to an \ac{MCS}.
Let $G_1 = (V_1, E_1)$ and $G_2 = (V_2, E_2)$ be a pair of graphs with labeled nodes/edges.
For the purpose of this work, we assume that graphs are directed by default.
We use $L(v)$ to denote the label of some node $v$.
For example, assuming $v$ is a node of a logic flow, $L(v)$ can be something as simple as the node's type (e.g. \ActFlowInstruction).
Given some label $\ell$, we use $V_i^\ell$ to denote the subset of nodes $v \in V_i$ such that $L(v) = \ell$.
Analogously, we use $L(u, v)$ to denote the label of some edge $(u, v)$ and $E_i^\ell$ to denote the subset of edges $(u, v) \in E_i$ such that $L(u, v) = \ell$.
For convenience, we use $L_{comb}(u, v) = (L(u), L(u, v), L(v))$ to denote the combined label of $(u, v)$ and $E_i^{comb / \ell}$ to denote the subset of edges $(u, v) \in E_i$ such that $L_{comb}(u, v) = \ell$.
Also, we abuse notation and use $L_{comb}(E_i)$ to denote the set of combined labels that occur in $E_i$.

A graph $G_C = (V_C, E_C)$ is a common sub-graph of $G_1$ and $G_2$ if there exist mappings $f_1: V_C \rightarrow V_1$ and $f_2: V_C \rightarrow V_2$ such that $L(v) = L(f_1(v)) = L(f_2(v))$ for all $v \in V_C$ and $L(u, v) = L(f_1(u), f_1(v)) = L(f_2(u), f_2(v))$ for all $(u, v) \in E_C$.
$G_C$ is said to be an \ac{MCS} if and only if no common sub-graph $G_C' = (V_C', E_C')$ of $G_1$ and $G_2$ exists containing more nodes or edges than $G_C$, i.e. such that $\left| V_C' \right| > \left| V_C \right|$ or $\left| E_C' \right| > \left| E_C \right|$.
For convenience, given a node $v \in V_i$, we abuse notation and use $v \in V_C$ to denote that there exists $v' \in V_C$ such that $v'$ is mapped to $v$, i.e. $f_i(v') = v$.
Analogously, given $(u, v) \in E_i$, we use $(u, v) \in E_C$ to denote that there exists $(u', v') \in E_C$ such that $f_i(u') = u$ and $f_i(v') = v$.

\subsection{Maximum Satisfiability}
\label{sec:bg-maxsat}

\ac{MCS} computation is well-known to be an NP-hard problem.
In recent years, \ac{MaxSAT} solvers have become a very effective tool for solving such hard combinatorial optimization problems~\cite{maxsat-survey,maxsat-ols,maxsat-maxhs,maxsat-res-comms}, thus our approach reduces the \ac{MCS} problem to \ac{MaxSAT}.

Let $X$ be a set of Boolean variables.
A literal $l$ is either a variable $x \in X$ or its negation $\neg x$.
A clause $c$ is a disjunction of literals $(l_1 \vee l_2 \vee \dots \vee l_k)$.
If a clause contains a single literal, then it is said to be a unit clause.
A propositional logic formula in \ac{CNF} $\phi$ is a conjunction of clauses $c_1 \wedge c_2 \wedge \dots \wedge c_n$.
A literal $x$ ($\neg x$) is said to be satisfied if and only if $x$ is assigned the Boolean value 1 (0).
A clause is satisfied if and only if at least one of its literals is satisfied.
A \ac{CNF} formula is satisfied if and only if all of its clauses are satisfied.
Given a \ac{CNF} formula $\phi$, the \ac{SAT} problem consists of deciding if there exists an assignment $\alpha: X \rightarrow \lbrace 0, 1 \rbrace$ of Boolean values to the variables of $X$ that satisfies $\phi$.
If $\alpha$ exists, then $\alpha$ is said to be a model of $\phi$.
Otherwise, $\phi$ is said to be unsatisfiable.
%Given a variable $x \in X$, we use $\alpha(x)$ to denote the value assigned to $x$ by assignment $\alpha$.

\ac{MaxSAT}~\cite{maxsat} is a generalization of \ac{SAT} where, in addition to the CNF formula $\phi$ (referred to as the hard formula), we have a set $S$ of soft clauses.
The goal is to compute a model $\alpha$ of $\phi$ that minimizes the number of clauses in $S$ not satisfied by $\alpha$.

\begin{example}
Consider the \ac{MaxSAT} instance with hard formula $\phi = \left( \neg x_1 \vee x_2 \right)$ and soft clauses $S = \lbrace \left( x_1 \right), \left( \neg x_2 \right) \rbrace$.
The assignment $\lbrace \left( x_1, 1 \right), \left( x_2, 0 \right) \rbrace$ is not a model of $\phi$.
On the other hand, the assignment $\lbrace \left( x_1, 0 \right), \left( x_2, 0 \right) \rbrace$ is a model of $\phi$ that satisfies the soft clause $\left( \neg x_2 \right)$.
Additionally, it is an optimal model since it is not possible to satisfy more than 1 soft clause for this instance.
\end{example}

\section{Single Pattern Extraction}
\label{sec:spe}

In order to mine duplicated code patterns, one must be able to extract a maximal pattern from a pair of logic flows $G_1 = (V_1, E_1)$ and $G_2 = (V_2, E_2)$.
The maximal pattern is an \ac{MCS} of $G_1$ and $G_2$.
Our approach reduces the problem of finding such an \ac{MCS} to an instance of \ac{MaxSAT}.
The \ac{MaxSAT} encoding is presented in Section~\ref{sec:spe-encoding}.
Section~\ref{sec:spe-preprocess} follows with an explanation of several pre-processing rules used to simplify $G_1$ and $G_2$ before building the encoding.

\subsection{MaxSAT Encoding}
\label{sec:spe-encoding}

Our \ac{MaxSAT} formulation is inspired by previous work on malware signature synthesis using \ac{MaxSAT}~\cite{malware-sig-synth}.
It extracts an \ac{MCS} by mapping the nodes of $G_2$ into the nodes of $G_1$.
The encoding is explained through a running a example in which we consider $G_1$ and $G_2$ to be the logic flows in Figures~\ref{fig:act-flow-simple} and~\ref{fig:act-flow-foreach} respectively.
Note that some mappings are not valid, such as mapping an \ActFlowIf~node to an \ActFlowInstruction.
In order to specify such constraints, node and edge labels are used.
In the example, the node/edge types are considered as labels for ease of explanation.
Additionally, the \ActFlowStart and \ActFlowEnd nodes must appear in every flow, and thus cannot be refactored to a separate flow.
Therefore, such nodes are discarded beforehand.

The following three sets of Boolean variables are considered:
\begin{itemize}
    \item \textbf{Inclusion variables.} For each node $v \in V_1$, a variable $o_v$ is introduced to encode if $v$ is part of the \ac{MCS} (i.e. $o_v = 1$) or not (i.e. $o_v = 0$).
    In the running example, three inclusion variable are needed: $o_{trim}$, $o_{low}$ and $o_{rep}$.
    \item \textbf{Mapping variables.} For each node pair $v, v'$ such that $v \in V_1$ and $v' \in V_2$, a variable $f_{v, v'}$ is introduced to encode if $v'$ is mapped to $v$ (i.e. $f_{v, v'} = 1$) or not (i.e. $f_{v, v'} = 0$).
    In the example, five variables are needed for each node of $G_1$. For the ToLower node, these variables are: $f_{low, for}$, $f_{low, trim}$, $f_{low, low}$, $f_{low, rep}$ and $f_{low, list}$.
    \item \textbf{Control-flow variables.} For each edge $(u, v) \in E_1$, a variable $c_{u, v}$ is introduced to encode if $(u, v)$ is part of the \ac{MCS} (i.e. $c_{u, v} = 1$) or not (i.e. $c_{u, v} = 0$).
    In the example, two control-flow variables are needed: $c_{trim, low}$ and $c_{low, rep}$.
\end{itemize}

For ease of explanation, some constraints are shown as at-most-1 constraints, i.e. of the form $\sum_i l_i \le 1$, instead of clauses.
Note that these are easily convertible to \ac{CNF} by introducing the clause $(\neg l_i \vee \neg l_j)$ for each pair $i, j$ such that $i \neq j$.
The hard formula contains the following constraints:
\begin{itemize}
    \item \textbf{Inclusion clauses.} A node $v \in V_1$ is in the \ac{MCS} if and only if at least one node in $V_2$ is mapped to $v$. If $v$ is the ToLower node, we have:
    \begin{multline}
        (\neg o_{low} \vee f_{low, for} \vee \dots \vee f_{low, list}) \wedge \\ (o_{low} \vee \neg f_{low, for}) \wedge \dots \wedge (o_{low} \vee \neg f_{low, list}) \text{.}
    \end{multline}
    \item \textbf{One-to-one clauses.} At most one node in $V_2$ can be mapped to each node $v \in V_1$. Assuming that $v$ is the ToLower node:
    \begin{equation}
        f_{low, for} + f_{low, trim} + f_{low, low} + f_{low, rep} + f_{low, list} \le 1 \text{.}
    \end{equation}
    \item \textbf{Function property clauses.} Each node $v' \in V_2$ cannot be mapped to more than one node in $V_1$. If $v'$ is the \ActFlowLoop node, we have:
    \begin{equation}
        f_{trim, for} + f_{low, for} + f_{rep, for} \le 1 \text{.} 
    \end{equation}
    \item \textbf{Label consistency clauses.} A node $v' \in V_2$ cannot be mapped to $v \in V_1$ if $v$ and $v'$ do not share the same label:
    \begin{equation}
        (\neg f_{trim, for}) \wedge (\neg f_{trim, list}) \wedge \dots \wedge (\neg f_{rep, for}) \wedge (\neg f_{rep, list}) \text{.}
    \end{equation}
    \item \textbf{Control-flow consistency clauses.} Consider some edge $(u, v) \in E_1$ and a pair of nodes $u', v' \in V_2$. If $u'$ and $v'$ are mapped to $u$ and $v$ respectively, and $(u', v')$ is not an edge of $G_2$ or does not share the same label as $(u, v)$, then $(u, v)$ cannot be in the \ac{MCS}.
    For example, if $u$ and $v$ are the ToLower and Replace nodes of $G_1$ respectively, since an edge does not exist between the ToLower and Trim of $G_2$, the following constraint is necessary:
    \begin{equation}
        (\neg f_{low, low} \vee \neg f_{rep, trim} \vee \neg c_{low, rep}) \text{.}
    \end{equation}
    On the other hand, the same constraint is not added when $u'$ and $v'$ are the Replace and ListAppend nodes of $G_2$ respectively, since the edge exists in $G_2$ and shares the same label as the edge between the ToLower and Replace of $G_1$.
    \item \textbf{No spurious edge clauses.} An edge $(u, v) \in E_1$ can be part of the \ac{MCS} only if both $u$ and $v$ are as well.
    If $u$ and $v$ are the ToLower and Replace nodes:
    \begin{equation}
        (\neg c_{trim, low} \vee o_{trim}) \wedge (\neg c_{trim, low} \vee o_{low}) \text{.}
    \end{equation}
    \item \textbf{No isolate node clauses.} A node $v \in V_1$ can be part of the \ac{MCS} only if at least one of its incoming/outgoing edges is in the \ac{MCS}. Assuming that $v$ is the ToLower node:
    \begin{equation}
        (\neg o_{low} \vee c_{trim, low} \vee c_{low, rep}) \text{.}
    \end{equation}
\end{itemize}

Note that the definition of \ac{MCS} provided in Section~\ref{sec:bg-mcs} does not forbid the inclusion of isolate nodes.
However, this is forbidden by the hard formula because such nodes are not desirable for the duplicated code pattern mining use case.

The optimization goal is to maximize the number of edges in the \ac{MCS}, which is given by the following set of soft clauses:
\begin{equation}\label{eq:spe-encoding-soft}
    \left\{ (c_{trim, low}), (c_{low, rep}) \right\} \text{.}
\end{equation}

Although the encoding described here focuses on extracting an \ac{MCS} of a pair of graphs, it can be easily extended to $k$ graphs by considering $k-2$ extra sets of mapping variables and adding the respective constraints to the hard formula.

\subsection{Graph Pre-processing}
\label{sec:spe-preprocess}

The pattern mining algorithms described in Section~\ref{sec:mining} rely on solving several \ac{MCS} instances.
Therefore, \ac{MCS} extraction must be as efficient as possible, since its performance strongly impacts the performance of the pattern miner.
\ac{MCS} instances can become hard to solve as the size of $G_1$ and $G_2$ increases.
For this reason, several pre-processing rules were implemented in order to reduce the size of $G_1$ and $G_2$.
The first rule discards edges with combined labels that do not occur in both $E_1$ and $E_2$, since it is impossible for an edge to be in the pattern if it does not occur in both graphs.
For the running example from Figures~\ref{fig:act-flow-simple} and~\ref{fig:act-flow-foreach}, this corresponds to discarding the edges that contain the \ActFlowLoop and ListAppend nodes.

\begin{proposition}\label{prop:spe-preprocess-first-rule}
    Given a pair of graphs $G_1 = (V_1, E_1)$ and $G_2 = (V_2, E_2)$, and an edge $(u, v) \in E_1$ such that $L_{comb}(u, v) \notin L_{comb}(E_2)$, then an \ac{MCS} of $G_1$ and $G_2$ is also an \ac{MCS} of $G_1'$ and $G_2$, where $V_1' = V_1$ and $E_1' = E_1 \setminus \lbrace (u, v) \rbrace$, and vice-versa.
\end{proposition}

\begin{proof}
    Let $G_C = (V_C, E_C)$ be an \ac{MCS} of $G_1$ and $G_2$.
    If $G_C$ is not an \ac{MCS} of $G_1'$ and $G_2$, then $(u, v) \in E_C$ since it is the only edge of $E_1$ not in $E_1'$.
    However, because $L_{comb}(u, v) \notin L_{comb}(E_2)$, no edge $(u', v') \in E_2$ exists such that $L(u) = L(u')$, $L(v) = L(v')$ and $L(u, v) = L(u', v')$, and thus, by definition, $(u, v)$ cannot be in $E_C$, resulting in a contradiction.
    On the other hand, if $G_C$ is an \ac{MCS} of $G_1'$ and $G_2$ but not of $G_1$ and $G_2$, then there must exist edges $(p, q) \in E_1 \setminus E_1'$ and $(p', q') \in E_2$ such that $L_{comb}(p, q) = L_{comb}(p', q')$.
    By definition, $E_1 \setminus E_1' = \lbrace (u, v) \rbrace$, thus $L_{comb}(p, q) = L_{comb}(u, v)$ which implies that $L_{comb}(p, q) \notin L_{comb}(E_2)$, hence $(p', q')$ does not exist.
\end{proof}

The application of Proposition~\ref{prop:spe-preprocess-first-rule} may cause either $G_1$ or $G_2$ to become disconnected.
More specifically, some edges may become what we refer to as orphan edges, i.e. an edge $(u, v) \in E_i$ such that $u$ and $v$ do not appear in any edges of $E_i$ other than $(u, v)$.
In other words, no other edge $(p, q) \in E_i$ exists such that ${p \in \lbrace u, v \rbrace}$ or ${q \in \lbrace u, v \rbrace}$.
Let $O_i^{comb / \ell}$ denote the subset of orphan edges in $E_i^{comb / \ell}$.
If $\left| O_1^{comb / \ell} \right| > \left| E_2^{comb / \ell} \right|$, then $G_1$ is said to contain an excess of orphan edges with combined label $\ell$.
The second rule discards orphan edges responsible for excesses in $G_1$ and $G_2$ until this is no longer the case.
It is safe to do this because the \ac{MCS} can contain at most $\left| E_2^{comb / \ell} \right|$ edges with combined label $\ell$.
%It is safe to do this because it is impossible for more than $\left| E_2^{comb / \ell} \right|$ of the edges in $O_1^{comb / \ell}$ to be included in the \ac{MCS} and it is irrelevant which of those edges are chosen.

\begin{proposition}\label{prop:spe-preprocess-second-rule}
    Given a pair of graphs $G_1 = (V_1, E_1)$ and $G_2 = (V_2, E_2)$, and an orphan edge $(u, v) \in E_1$, if $G_1$ contains an excess of orphan edges with combined label $L_{comb}(u, v)$, then there exists an \ac{MCS} $G_C = (V_C, E_C)$ of $G_1$ and $G_2$ such that $(u, v) \notin E_C$.
\end{proposition}

\begin{proof}
    Let $G_C' = (V_C', E_C')$ be an \ac{MCS} of $G_1$ and $G_2$ such that $(u, v) \in E_C'$, and let $(p, q) \in E_C'$ be the edge of $G_C'$ such that $f_1'(p) = u$ and $f_1'(q) = v$.
    Because $(u, v)$ is an orphan edge, by definition $(p, q)$ must also be an orphan edge.
    Moreover, since $(u, v)$ is in excess, we have that $\left| O_1^{comb / L_{comb}(u, v)} \right| > \left| E_2^{comb / L_{comb}(u, v)} \right| \ge \left| E_C'^{comb / L_{comb}(u, v)} \right|$, and thus there exists at least one edge $(u', v') \in O_1^{comb / L_{comb}(u, v)}$ such that $(u', v') \notin E_C'$.
    Consequently, there exists a mapping $f_1$ identical to $f_1'$, with the exception that $f_1(p) = u'$ and $f_1(q) = v'$, thus $G_C$ exists.
\end{proof}

The aforementioned rules may also cause some of the connected components of some $G_i$ to become simple paths, i.e. a subgraph of $G_i$ with node set $V_S = \lbrace v_1, v_2, \dots, v_n \rbrace$ such that $(v_j, v_{j+1}) \in E_i$, for all $1 \le j < n$, and no other edge exists in $E_i$ with nodes from $V_S$.
Assuming $i = 1$, let $P_1^{(L_{comb}(v_1, v_2), \dots, L_{comb}(v_{n-1}, v_n))}$ denote the set of all simple path components $V_S' = \lbrace v_1', v_2', \dots, v_n' \rbrace$ in $G_1$ such that $L_{comb}(v_j, v_{j+1}) = L_{comb}(v_j', v_{j+1}')$ for all $1 \le j < n$.
The third rule discards $v_1$ ($v_n$) if there exist more components in $P_1^{(L_{comb}(v_1, v_2), \dots, L_{comb}(v_{n-1}, v_n))}$ than nodes in $V_2^{L(v_1)}$ ($V_2^{L(v_n)}$).
Similarly to Proposition~\ref{prop:spe-preprocess-second-rule}, this is allowed because the \ac{MCS} can contain at most $\left| V_2^{L(v_1)} \right|$ ($\left| V_2^{L(v_n)} \right|$) nodes with label $L(v_1)$ ($L(v_n)$).
We prove the correctness of this rule just for $v_1$, but note that the same reasoning applies for $v_n$.

\begin{proposition}\label{prop:spe-preprocess-third-rule}
    Given a pair of graphs $G_1 = (V_1, E_1)$ and $G_2 = (V_2, E_2)$ such that $G_1$ contains a simple path component $V_S =  \lbrace v_1, v_2, \dots, v_n \rbrace$, if $\left| P_1^{(L_{comb}(v_1, v_2), \dots, L_{comb}(v_{n-1}, v_n))} \right| > \left| V_2^{L(v_1)} \right|$, then there exists an \ac{MCS} $G_C = (V_C, E_C)$ of $G_1$ and $G_2$ such that $v_1 \notin V_C$.
\end{proposition}

\begin{proof}
    Let $G_C' = (V_C', E_C')$ be an MCS of $G_1$ and $G_2$ such that $v_1 \in V_C'$.
    We have that $\left| P_1^{(L_{comb}(v_1, v_2), \dots, L_{comb}(v_{n-1}, v_n))} \right| > \left| V_2^{L(v_1)} \right| \ge \left| V_C'^{L(v_1)} \right|$, thus there exists at least one simple path component $V_S' = \lbrace v_1', v_2', \dots, v_n' \rbrace$ in $P_1^{(L_{comb}(v_1, v_2), \dots, L_{comb}(v_{n-1}, v_n))}$ such that $v_1' \notin V_C'$.
    Without loss of generality, assume that $G_C'$ contains two simple path components $U_S = \lbrace u_1, u_2, \dots, u_n \rbrace$ and $U_S' = \lbrace u_2', u_3', \dots, u_n' \rbrace$ such that $f_1'(u_j) = v_j$ for all $1 \le j \le n$ and $f_1'(u_k') = v_k'$ for all $2 \le k \le n$.
    By definition, we have that $L_{comb}(u_j, u_{j+1}) = L_{comb}(v_j, v_{j+1}) = L_{comb}(v_j', v_{j+1}') = L_{comb}(u_j', u_{j+1}')$ for all $2 \le j < n$ and $L_{comb}(u_1, u_2) = L_{comb}(v_1, v_2) = L_{comb}(v_1', v_2')$.
    Therefore, there exists a mapping $f_1$ identical to $f_1'$, with the exception that $f_1(u_j) = f_1'(u_j')$ and $f_1(u_j') = f_1'(u_j)$ for all $2 \le j \le n$, and $f_1(u_1) = v_1'$, thus $G_C$ exists.
\end{proof}

The three rules are repeatedly used to simplify $G_1$ and $G_2$ until a fixpoint is reached, i.e. all the rules are no longer applicable.
At each iteration, isolate nodes are also discarded since our \ac{MaxSAT} encoding forbids the inclusion of such nodes in the \ac{MCS}, and doing so may enable further simplifications through Proposition~\ref{prop:spe-preprocess-third-rule}.

\section{Pattern Mining}
\label{sec:mining}

This section focuses on the problem of mining duplicated code patterns from a set of graphs $G_1, G_2, \dots, G_n$.
We start by describing a greedy pattern mining algorithm in Section~\ref{sec:mining-greedy}, followed by its lazy version in Section~\ref{sec:mining-lazy}.
In Section~\ref{sec:mining-isomorphic}, we propose an optimization that relies on de-duplicating the initial set of graphs before mining patterns.
Lastly, in Section~\ref{sec:mining-inv-index} we describe how an inverted index can be used in order to further reduce the algorithm's runtime.

\subsection{Greedy Algorithm}
\label{sec:mining-greedy}

We propose a pattern mining algorithm that follows a greedy approach.
The algorithm iteratively picks the graph pair $G, G'$ with the highest priority, according to some custom priority function, extracts a pattern $G_C$ of $G, G'$ and replaces $G, G'$ with $G_C$.
This process is repeated until there are no more graph pairs left to consider.

For the duplicated code use case, the priority function is based on the notion of refactor weight of a graph.
Given some graph $G = (V, E)$, each node $v \in V$ has an associated refactor weight $\omega_v$, which depends on its type and the kind of operations it performs.
We consider a refactor weight of $1$ for all nodes except \ActFlowInstruction~nodes that correspond to database accesses.
The weight of such nodes is given by the respective number of database tables, and filter and sort conditions.
Similarly, we consider a refactor weight of $\omega_{u, v} = 1$ for all edges $(u, v) \in E$.
Let $G_{W1} = (V_{W1}, E_{W1}), G_{W2} = (V_{W2}, E_{W2}), \dots, G_{Wp} = (V_{Wp}, E_{Wp})$ denote the $p$ weakly connected components of $G$.
A weakly connected component $G_{Wi}$ is a maximal sub-graph of $G$ such that, for all node pairs $u, v \in V_{Wi}$, $v$ is reachable from $u$ in the undirected counterpart of $G$.
The refactor weight $\omega_G$ of $G$ is given by:
\begin{equation}
    \omega_G = \max_{i \in \lbrace 1, 2, \dots, p \rbrace} \left\lbrace \sum_{v \in V_{Wi}} \omega_v + \sum_{(u, v) \in E_{Wi}} \omega_{u, v} \right\rbrace \text{.}
\end{equation}
We consider the maximum weight across $G$'s components instead of the sum because, from a refactoring perspective, patterns with less but bigger components are preferable.
Given a graph pair $G, G'$, its priority is an upper bound of the refactor weight of an \ac{MCS} of $G$ and $G'$.
Given two components $G_{Wi}, G_{Wj}'$ of $G, G'$ respectively, the upper bound $comp\_ub(G_{Wi}, G_{Wj}')$ for $G_{Wi}, G_{Wj}'$ is given by:
\begin{multline}
    \sum_{\ell \in L_{comb}(E_{Wi}) \cap L_{comb}(E_{Wj}')} \\ \min_{E_W \in \lbrace E_{W_i}^{comb/\ell}, E_{W_j}'^{comb/\ell} \rbrace} \left\lbrace \sum_{(u, v) \in E_W} \left( \omega_{u, v} + \omega_u + \omega_v \right) \right\rbrace \text{.}
\end{multline}
%\begin{multline}
%    \sum_{\ell \in L_{comb}(E_{Wi}) \cap L_{comb}(E_{Wj}')} \\ \min_{k \in \lbrace i, j \rbrace} \left\lbrace \sum_{(u, v) \in E_{Wk}^\ell} \left( \omega_{u, v} + \omega_u + \omega_v \right) \right\rbrace \text{.}
%\end{multline}
Assuming $G'$ has $q$ components, the refactor weight upper bound $ub(G, G')$ for $G$ and $G'$ is given by:
\begin{equation}\label{eq:mining-greedy-weight-ub}
    ub(G, G') = \max_{i, j \in \lbrace 1, 2, \dots, p \rbrace \times \lbrace 1, 2, \dots, q \rbrace} \left\lbrace comp\_ub(G_{Wi}, G_{Wj}') \right\rbrace \text{.}
\end{equation}
Ties in Equation~\eqref{eq:mining-greedy-weight-ub} are broken using an upper bound of the number of edges of an \ac{MCS} of $G$ and $G'$, given by:
\begin{equation}
    \sum_{\ell \in L_{comb}(E) \cap L_{comb}(E')} \min \left\lbrace \left| E^{comb/\ell} \right|, \left| E'^{comb/\ell} \right| \right\rbrace \text{.}
\end{equation}

%Let $\omega_{call}$ denote the refactor weight of an \ActFlowInstruction~node that represents a call to an action flow.
%Assuming that $G$ corresponds to a pattern that occurs over $N$ action flows, the refactor value of $G$ is given by:
%\begin{equation}
%    \omega_G \cdot \left( N - 1 \right) - \omega_{call} \cdot N \text{.}
%\end{equation}
%The rationale behind this equation is that, in the best case scenario, by refactoring $G$, one deletes nodes and edges from the code base that equate to a total refactoring value of $\omega_G \cdot N$.
%However, one must add a new action flow that implements the refactored logic (hence the $N - 1$ term) and replace the duplicated code with calls to the new action flow (hence the $\omega_{call} \cdot N$ term).

\begin{algorithm}[t]
    \SetKwFunction{Heapify}{Heapify}
    \SetKwFunction{Pop}{Pop}
    \SetKwFunction{Push}{Push}
    \SetKwFunction{ExtractPattern}{ExtractMCS}
    \SetKwFunction{PostProcessPattern}{PostProcessMCS}
    \KwIn{$G_1, G_2, \dots, G_n, \beta$}
    $R \leftarrow \emptyset$ \\
    $A \leftarrow \lbrace G_i: 1 \le i \le n \wedge \omega_{G_i} \ge \beta \rbrace$ \label{algl:mining-greedy-act-graphs-init} \\
    $Q \leftarrow \lbrace (ub(G_i, G_j), G_i, G_j): G_i, G_j \in A \wedge i \neq j \rbrace$ \label{algl:mining-greedy-queue-init} \\
    \Heapify{$Q$} \label{algl:mining-greedy-queue-heapify} \\
    \While{$\left| Q \right| > 0$}{ \label{algl:mining-greedy-loop-cond}
        $ub, G, G' \leftarrow \text{\Pop{$Q$}}$ \label{algl:mining-greedy-pop} \\
        \If{$ub \ge \beta \wedge G, G' \in A$}{ \label{algl:mining-greedy-ub-active-check}
            $G_C \leftarrow \text{\ExtractPattern{$G, G'$}}$ \label{algl:mining-greedy-extract-pattern} \\
            %$G_C \leftarrow \text{\PostProcessPattern{$G_C$}}$ \label{algl:mining-greedy-post-process} \\
            \If{$\omega_{G_C} \ge \beta$}{ \label{algl:mining-greedy-pattern-thres-check}
                $A \leftarrow A \setminus \lbrace G, G' \rbrace$ \label{algl:mining-greedy-act-graphs-updt} \\
                $R \leftarrow R \cup \lbrace G_C \rbrace$ \label{algl:mining-greedy-store} \\
                \ForEach{$G \in A$}{ \label{algl:mining-greedy-queue-updt-loop-cond}
                    \Push{$Q, (ub(G, G_C), G, G_C)$} \label{algl:mining-greedy-queue-updt-body}
                }
                $A \leftarrow A \cup \lbrace G_C \rbrace$ \label{algl:mining-greedy-act-graphs-pattern-updt} \\
            }
        }
    }
    \KwRet{$R$}
    \caption[Greedy pattern mining algorithm]{Greedy pattern mining algorithm.}
    \label{alg:mining-greedy}
\end{algorithm}

The greedy pattern mining algorithm is presented in Algorithm~\ref{alg:mining-greedy}.
It receives as input a set of $n$ graphs $G_1, G_2, \dots, G_n$ and a minimum refactor weight threshold $\beta$, and returns a set $R$ of maximal patterns with a refactor weight of at least $\beta$.
It starts by initializing a set $A$ of active graphs, discarding graphs with a refactor weight lower than $\beta$ (line~\ref{algl:mining-greedy-act-graphs-init}).
Then, it initializes a priority queue $Q$ with all possible pairs of graphs in $A$ (lines~\ref{algl:mining-greedy-queue-init} and~\ref{algl:mining-greedy-queue-heapify}).
While $Q$ is not empty (line~\ref{algl:mining-greedy-loop-cond}), it repeatedly pops a pair $G$ and $G'$ from the queue (line~\ref{algl:mining-greedy-pop}), and, if the upper bound for $G$ and $G'$ satisfies the threshold $\beta$ and both graphs are still active (line~\ref{algl:mining-greedy-ub-active-check}), it extracts an \ac{MCS} $G_C$ of $G$ and $G'$ using the approach described in Section~\ref{sec:spe} (line~\ref{algl:mining-greedy-extract-pattern}).
If the refactor weight of $G_C$ satisfies the threshold $\beta$ (line~\ref{algl:mining-greedy-pattern-thres-check}), then $G$ and $G'$ are removed from the active set $A$ (line~\ref{algl:mining-greedy-act-graphs-updt}), $G_C$ is stored in $R$ (line~\ref{algl:mining-greedy-store}), new pairs with $G_C$ and the remaining active graphs are added to $Q$ (lines~\ref{algl:mining-greedy-queue-updt-loop-cond} and~\ref{algl:mining-greedy-queue-updt-body}), and $G_C$ is added to the active graph set (line~\ref{algl:mining-greedy-act-graphs-pattern-updt}).

%Note that the algorithm allows some custom post-processing of the \acp{MCS} after their extraction (line~\ref{algl:mining-greedy-post-process}).
%This is supported for the following reasons:
%\begin{inlineenum}
%    \item It may be the case that $G_C$ contains some \ActFlowIf~or \ActFlowSwitch~node $v$ with none of its branches in the \ac{MCS}, i.e. no $(u', v') \in E_C$ exists such that $u' = v$.
%    Such nodes cannot be refactored to a separate logic flow, thus we discard them and the respective edges in post-processing.
%    \item Even though the refactoring weight of $G_C$ may satisfy the threshold $\beta$, it may be the case that some of its weakly connected components do not.
%    Such components are discarded as well during post-processing.
%\end{inlineenum}

Due to its greedy nature, one can extend the algorithm in order to obtain a tree hierarchy of the patterns.
Let $G$ and $G'$ be duplicated code patterns that occur across the logic flows in sets $F$ and $F'$ respectively.
Assuming that, at some point during its execution, the algorithm extracts an \ac{MCS} $G_C$ for $G$ and $G'$, then $G_C$ is a possibly smaller pattern that occurs across the flows in $F \cup F'$.
The tree hierarchy would contain an internal node for $G_C$ with two children nodes for $G$ and $G'$.
Analogously, children of $G$ would represent possibly larger patterns that occur in subsets of $F$.
In the future, we plan to explore ways of exploiting this tree hierarchy in order to provide a guided refactoring experience to the user.

\subsection{Lazy Greedy Algorithm}
\label{sec:mining-lazy}

Recall that the pattern miner must return a response within a given time budget.
If said budget expires, the miner should still return a subset of maximal patterns.
Algorithm~\ref{alg:mining-greedy} may incur a long delay until the first pattern extraction due to the eager initialization of the priority queue (line~\ref{algl:mining-greedy-queue-init}), which requires computing refactor weight upper bounds for $O(n^2)$ candidate graph pairs.
For example, for one of our test code bases, the pattern miner must handle about $13$K flows, which corresponds to almost $85$M pairs.
Queue initialization can take up to a couple of hours for such large code bases.

To solve this issue, we propose a lazy version of Algorithm~\ref{alg:mining-greedy}, based on the observation that, given a graph pair $G_i, G_j$, $1 \le i, j \le n$, such that $i \neq j$, and $ub(G_i, G_j) \ge ub(G_i, G_k)$ and $ub(G_i, G_j) \ge ub(G_j, G_k)$ for all $1 \le k \le n$, where $ub$ is the refactor weight upper bound from Equation~\eqref{eq:mining-greedy-weight-ub}, then we can safely extract an \ac{MCS} for $G_i$ and $G_j$ before performing any further upper bound computations.
This property comes as a consequence of the monotonicity of $ub$.

\begin{proposition}
    Given three graphs $G_1 = (V_1, E_1)$, $G_2 = (V_2, E_2)$ and $G_3 = (V_3, E_3)$, and an \ac{MCS} $G_C = (V_C, E_C)$ of $G_1$ and $G_2$, we have that $ub(G_1, G_3) \ge ub(G_C, G_3)$ and $ub(G_2, G_3) \ge ub(G_C, G_3)$.
\end{proposition}

\begin{proof}
    Without loss of generality, assume that $G_1$, $G_2$, $G_3$ and $G_C$ contain a single weakly connected component.
    By definition, $G_C$ is a sub-graph of $G_1$.
    Consequently, we have that $L_{comb}(E_C) \subseteq L_{comb}(E_1)$, implying that $L_{comb}(E_C) \cap L_{comb}(E_3) \subseteq L_{comb}(E_1) \cap L_{comb}(E_3)$.
    Additionally, we have that $E_C^\ell \subseteq E_1^\ell$ for any label $\ell$, thus:
    \begin{multline}
        ub(G_C, G_3) = \\ = \sum_{\ell \in L_{comb}(E_C) \cap L_{comb}(E_3)} \min_{E \in \lbrace E_C, E_3 \rbrace} \left\lbrace \sum_{(u, v) \in E^\ell} \left( \omega_{u, v} + \omega_u + \omega_v \right) \right\rbrace \le \\ \le \sum_{\ell \in L_{comb}(E_1) \cap L_{comb}(E_3)} \min_{E \in \lbrace E_1, E_3 \rbrace} \left\lbrace \sum_{(u, v) \in E^\ell} \left( \omega_{u, v} + \omega_u + \omega_v \right) \right\rbrace = \\ = ub(G_1, G_3) \text{.}
    \end{multline}
    Same goes for $G_2$.
\end{proof}

\begin{algorithm}[t]
    \DontPrintSemicolon
    \SetKwProg{Fn}{Function}{}{end function}
    \SetKwFunction{ActivateGraph}{ActivateGraph}
    \SetKwFunction{First}{First}
    \SetKwFunction{Push}{Push}
    \SetKwFunction{Pop}{Pop}
    \KwIn{$G_1, G_2, \dots, G_n, \beta$}
    \;
    \Fn{\ActivateGraph{$Q, A, I, G$}}{
        $I \leftarrow I \setminus \lbrace G \rbrace$ \label{algl:mining-lazy-activate-remove-inactive} \\
        \ForEach{$G' \in I$}{ \label{algl:mining-lazy-activate-loop-cond}
            \Push{$Q, (ub(G, G'), G, G')$} \label{algl:mining-lazy-activate-loop-body}
        }
        $A \leftarrow A \cup \lbrace G \rbrace$ \label{algl:mining-lazy-activate-add-active} \\
        \KwRet{$Q, A, I$}
    }
    \;
    $R, A, Q \leftarrow \emptyset, \emptyset, \emptyset$ \label{algl:mining-lazy-empty-set-init} \\
    $I \leftarrow \lbrace G_i: 1 \le i \le n \wedge \omega_{G_i} \ge \beta \rbrace$ \label{algl:mining-lazy-inactive-init} \\
    \While{$\left| Q \right| > 0 \vee \left| I \right| > 1$}{ \label{algl:mining-lazy-main-loop-cond}
        \If{$\left| Q \right| = 0$}{ \label{algl:mining-lazy-empty-queue-check}
            $Q, A, I \leftarrow \ActivateGraph{Q, A, I, \First{I}}$ \label{algl:mining-lazy-activate-first}
        }
        $ub, G, G' \leftarrow \First{Q}$ \label{algl:mining-lazy-act-loop-init} \\
        \While(\tcp*[h]{assume $G \notin I$ for simplicity}){$G' \in I$}{
            $Q, A, I \leftarrow \ActivateGraph{Q, A, I, G'}$ \\
            $ub, G, G' \leftarrow \First{Q}$ \label{algl:mining-lazy-act-loop-end}
        }
        $\Pop{Q}$ \label{algl:mining-lazy-queue-pop} \\
        \If{$ub \ge \beta \wedge G, G' \in A$}{
            $G_C \leftarrow \text{\ExtractPattern{$G, G'$}}$ \\
            %$G_C \leftarrow \text{\PostProcessPattern{$G_C$}}$ \\
            \If{$\omega_{G_C} \ge \beta$}{
                $A \leftarrow A \setminus \lbrace G, G' \rbrace$ \\
                $R \leftarrow R \cup \lbrace G_C \rbrace$ \\
                \ForEach{$G \in A$}{ \label{algl:mining-lazy-pattern-updt-loop-cond}
                    \Push{$Q, (ub(G, G_C), G, G_C)$} \label{algl:mining-lazy-pattern-updt-loop-body}
                }
                $I \leftarrow I \cup \lbrace G_C \rbrace$ \label{algl:mining-lazy-pattern-inactive-updt}
            }
        }
    }
    \KwRet{$R$}
    \caption[Lazy greedy pattern mining algorithm]{Lazy greedy pattern mining algorithm.}
    \label{alg:mining-lazy}
\end{algorithm}

The lazy greedy pattern mining algorithm is presented in Algorithm~\ref{alg:mining-lazy}.
It shares many similarities with Algorithm~\ref{alg:mining-greedy}, the main difference being the management of the priority queue $Q$ and active graph set $A$.
Initially, $Q$ and $A$ are empty (line~\ref{algl:mining-lazy-empty-set-init}) and a set of inactive graphs $I$ is initialized with all graphs with a refactor weight that satisfies the threshold $\beta$ (line~\ref{algl:mining-lazy-inactive-init}).
At each iteration, the algorithm starts by checking if $Q$ is empty (line~\ref{algl:mining-lazy-empty-queue-check}).
If so, then a graph $G \in I$ is activated (line~\ref{algl:mining-lazy-activate-first}).
This corresponds to moving $G$ from $I$ to $A$ (lines~\ref{algl:mining-lazy-activate-remove-inactive} and~\ref{algl:mining-lazy-activate-add-active}) and adding new pairs to $Q$ containing $G$ and each remaining inactive graph (lines~\ref{algl:mining-lazy-activate-loop-cond} and~\ref{algl:mining-lazy-activate-loop-body}).
Next, if necessary, additional graphs are activated until the pair in $Q$ with the highest upper bound no longer contains inactive graphs (lines~\ref{algl:mining-lazy-act-loop-init} to~\ref{algl:mining-lazy-act-loop-end}).
The rest of the algorithm (lines~\ref{algl:mining-lazy-queue-pop} to~\ref{algl:mining-lazy-pattern-inactive-updt}) behaves in the same way as Algorithm~\ref{alg:mining-greedy}, with the exception that each new \ac{MCS} $G_C$ is added to the inactive set $I$ instead of $A$ (line~\ref{algl:mining-lazy-pattern-inactive-updt}).
This process is repeated until $Q$ becomes empty and at most $1$ inactive graph is left (line~\ref{algl:mining-lazy-main-loop-cond}).

%This property comes as a consequence of the monotonicity of the refactor weight of a graph.

%\begin{proposition}
%    Given a pair of graphs $G_1 = (V_1, E_1)$ and $G_2 = (V_2, E_2)$, and an \ac{MCS} $G_C = (V_C, E_C)$ of $G_1$ and $G_2$, we have that $\omega_{G_1} \ge \omega_{G_C}$ and $\omega_{G_2} \ge \omega_{G_C}$.
%\end{proposition}

%\begin{proof}
%    Without loss of generality, assume that $G_1$, $G_2$ and $G_C$ contain a single weakly connected component.
%    By definition, $G_C$ is a sub-graph of $G_1$, thus:
%    \begin{equation}
%        \omega_{G_C} = \sum_{v \in V_C} \omega_v + \sum_{(u, v) \in E_C} \omega_{u, v} \le \sum_{v \in V_1} \omega_v + \sum_{(u, v) \in E_1} \omega_{u, v} = \omega_{G_1} \text{.}
%    \end{equation}
%    Same goes for $G_2$.
%\end{proof}

\subsection{Isomorphic Logic Flows}
\label{sec:mining-isomorphic}

In practice, we observed that it is common for some of the logic flows to be fully duplicated.
For example, among the $13$K flows in the test code base mentioned at the start of the previous section, about $2$K of them ($\approx 15\%$) are full duplicates.
Finding full duplicates is much cheaper than mining \acp{MCS}, hence we propose an algorithm for de-duplicating the code base in order to significantly reduce the number of refactor weight upper bound computations.

Given two graphs $G_1 = (V_1, E_1)$ and $G_2 = (V_2, E_2)$ and an \ac{MCS} $G_C = (V_C, E_C)$ of $G_1$ and $G_2$, we say that $G_1$ and $G_2$ are isomorphic if and only if, for all $v \in V_1 \cup V_2$, $v \in V_C$, and for all $(u, v) \in E_1 \cup E_2$, $(u, v) \in E_C$.
When this is the case, we refer to $G_C$ as an isomorphic duplicated code pattern.
The \ac{MaxSAT} encoding presented in Section~\ref{sec:spe-encoding} can be adapted to extract only isomorphic patterns by adding the following hard clauses:
\begin{itemize}
    \item Unit clauses containing each of the inclusion and control-flow variables.
    \item A clause $\left( \bigvee_{v \in V_1} f_{v, v'} \right)$ for each node $v' \in V_2$.
    \item A clause $(\neg f_{u, u'} \vee \neg f_{v, v'})$ for each edge $(u', v') \in E_2$ and nodes $u, v \in V_1$ such that $(u, v) \notin E_1$ or $L(u, v) \neq L(u', v')$.
\end{itemize}
This variant is a decision problem, which can be solved much more efficiently than its optimization version.
In fact, in many practical scenarios, one can quickly conclude that $G_1$ and $G_2$ are not isomorphic by checking if $\left| V_1 \right| \neq \left| V_2 \right|$ or $\left| E_1 \right| \neq \left| E_2 \right|$, or if any of the pre-processing rules in Section~\ref{sec:spe-preprocess} is applicable.

\begin{algorithm}[t]
    \DontPrintSemicolon
    \SetKw{break}{break}
    \SetKwFunction{Sort}{Sort}
    \SetKwFunction{IsIsomorphic}{IsIsomorphic}
    \SetKwFunction{GetIsomorphicPattern}{GetIsomorphicPattern}
    \SetKwFunction{GetPatterns}{GetPatterns}
    \KwIn{$G_1, G_2, \dots, G_n$}
    $D \leftarrow \emptyset$ \label{algl:mining-iso-init-dict} \\
    \For{$i \leftarrow 1$ \KwTo $n$}{
        $key \leftarrow \text{\Sort{$[L_{comb}(u, v): (u, v) \in E_i]$}}$ \label{algl:mining-iso-key} \\
        \ForEach{$G \in D[key]$}{ \label{algl:mining-iso-inner-loop-cond}
            \If{\IsIsomorphic{$G_i, G$}}{ \label{algl:mining-iso-iso-check}
                $G_C \leftarrow \text{\GetIsomorphicPattern{$G_i, G$}}$ \label{algl:mining-iso-pattern} \\
                $D \leftarrow \left( D \setminus \lbrace (key, G) \rbrace \right) \cup \lbrace (key, G_C) \rbrace$ \label{algl:mining-iso-dict-pattern-updt} \\
                \break
            }
        }
        \If{$\nexists_{G \in D[key]} \, \text{\IsIsomorphic{$G_i, G$}}$}{ \label{algl:mining-iso-no-iso-check}
            $D \leftarrow D \cup \lbrace (key, G_i) \rbrace$ \label{algl:mining-iso-dict-graph-updt}
        }
    }
    \KwRet{\GetPatterns{$D$}} \label{algl:mining-iso-return}
    \caption[Isomorphic pattern mining algorithm]{Isomorphic pattern mining algorithm.}
    \label{alg:mining-iso}
\end{algorithm}

The isomorphic pattern mining algorithm is presented in Algorithm~\ref{alg:mining-iso}.
It maintains a dictionary $D$ of lists of graphs where the isomorphic patterns are stored.
Initially, $D$ is empty (line~\ref{algl:mining-iso-init-dict}).
For each graph $G_i$, the algorithm starts by computing the key for $G_i$, which is the sorted concatenation of the combined labels of the edges in $E_i$ (line~\ref{algl:mining-iso-key}).
Next, it checks if there exists a graph $G$ in $D$ with the same key as $G_i$, such that $G$ and $G_i$ are isomorphic (lines~\ref{algl:mining-iso-inner-loop-cond} and~\ref{algl:mining-iso-iso-check}).
Note that $G$ and $G_i$ will have the same key if and only if each combined label appears the exact same number of times in both graphs, which is a necessary condition in order for $G$ and $G_i$ to be isomorphic.
If such $G$ exists, then an isomorphic pattern $G_C$ is extracted for $G$ and $G_i$ (line ~\ref{algl:mining-iso-pattern}), and $G$'s entry in $D$ is replaced with $G_C$ (line~\ref{algl:mining-iso-dict-pattern-updt}).
Otherwise, $G_i$ is added to $D$ (lines~\ref{algl:mining-iso-no-iso-check} and~\ref{algl:mining-iso-dict-graph-updt}).
Finally, the isomorphic patterns in $D$ are returned by the algorithm (line~\ref{algl:mining-iso-return}).

\subsection{Inverted Index}
\label{sec:mining-inv-index}

Although de-duplication helps, the number of candidate graph pairs can still be prohibitively high.
For example, de-duplicating the test code base reduces the number of flows to $11$K, which still results in about $61$M pairs. 
In order to further reduce this number, we use a partial inverted index like the one proposed in \SourcererCC~\cite{sourcerercc}.
In the context of our work, the inverted index is a mapping of combined edge labels to lists of graphs that those labels appear in.
The index is deemed partial because it contains entries only for a subset of combined labels that occur with the most frequency.

\begin{algorithm}[t]
    \DontPrintSemicolon
    \SetKw{break}{break}
    \SetKwFunction{SortByGlobalFrequency}{SortByGlobalFrequency}
    \KwIn{$G_1, G_2, \dots, G_n, \delta$}
    $I \leftarrow \emptyset$ \\
    \For{$i \leftarrow 1$ \KwTo $n$}{
        $B \leftarrow \text{\SortByGlobalFrequency{$L_{comb}(E_i)$}}$ \label{algl:mining-inv-index-sort} \\
        \For{$j \leftarrow 1$ \KwTo $\lceil \left| B \right| \cdot \delta \rceil$}{ \label{algl:mining-inv-index-updt-loop}
            $I \leftarrow I \cup \lbrace (B[j], G_i) \rbrace$ \label{algl:mining-inv-index-updt}
        }
    }
    \KwRet{$I$}
    \caption[Partial inverted index creation]{Partial inverted index creation.}
    \label{alg:mining-inv-index}
\end{algorithm}

Algorithm~\ref{alg:mining-inv-index} describes the process of creating the index for a given set of graphs $G_1, G_2, \dots, G_n$.
For each graph $G_i$, it starts by creating a bag $B$ of the combined labels that appear in $E_i$, sorted in decreasing order of their global frequency (line~\ref{algl:mining-inv-index-sort}).
The global frequency of some combined label $\ell \in L_{comb}(E_i)$ is given by:
\begin{equation}
    \frac{\sum_{j = 1}^n \left| E_j^{comb/\ell} \right|}{\sum_{j = 1}^n \left| E_j \right|} \text{.}
\end{equation}
Lastly, entries containing $G_i$ are added to $I$ for a prefix of $B$ (lines~\ref{algl:mining-inv-index-updt-loop} and~\ref{algl:mining-inv-index-updt}).
The prefix size is controlled through the $\delta$ input parameter, which represents the fraction of a graph's combined labels to include in the index.
For example, if $\delta = 0.2$, then the $20\%$ most frequent combined labels in $B$ are included in $I$.

Algorithm~\ref{alg:mining-greedy} requires the following changes in order to integrate the inverted index:
(1) During queue initialization (line~\ref{algl:mining-greedy-queue-init}), only pairs of graphs that occur in the same index list are considered.
(2) A new pattern $G_C$ is added to the index before the queue update (lines~\ref{algl:mining-greedy-queue-updt-loop-cond} and~\ref{algl:mining-greedy-queue-updt-body}), and the respective new queue pairs should contain only graphs that occur in the same index lists as $G_C$.
The same reasoning applies to the queue updates in lines~\ref{algl:mining-lazy-activate-loop-cond}, \ref{algl:mining-lazy-activate-loop-body}, \ref{algl:mining-lazy-pattern-updt-loop-cond} and~\ref{algl:mining-lazy-pattern-updt-loop-body} of Algorithm~\ref{alg:mining-lazy}.

\section{Experimental Evaluation}
\label{sec:eval}

In this section, the performance of the pattern mining algorithms and optimizations proposed in Section~\ref{sec:mining} is evaluated.
The pattern miners were executed on benchmark sets of logic flows from a random sample of 800 real-world code bases written in \TargetVPL\footnote{Unfortunately, these code bases cannot be made publicly available due to client privacy agreements.}.
The 800 code bases were sampled uniformly from the full world of 1491 code bases that existed at experimentation time.
Note that an \TargetVPL code base contains multiple web/mobile applications, typically hundreds.
In order to protect sensitive data, the code bases are anonymized at the source, including the replacement of string literals with hashes. 
Two performance indicators are considered: mining time, i.e. elapsed time since the start of the mining algorithm until termination, and duplicated refactor weight found, i.e. total refactor weight of the nodes and edges that appear in the patterns returned by the algorithm.
Note that the same node/edge may appear multiple times across different patterns, but the respective refactor weight is counted only once.
Precision/recall is not considered because the respective results depend heavily on the node/edge labels, and the focus of this work is algorithm scalability.

\begin{table}[t]
    \centering
    \caption{Statistics on the number of flows and nodes per benchmark set.}
    \label{tb:inst-stats}
    \begin{tabular}{ l | r r r |}
        \textbf{Parameter} & \textbf{min} & \textbf{median} & \textbf{max} \\
        \hline
        Flows & $95$ & $2092$ & $45691$ \\
        Nodes & $908$ & $17043$ & $413135$ \\
        Flows considered & $20$ & $607$ & $15463$ \\
        Nodes considered & $123$ & $5228$ & $144213$ \\
        \hline
    \end{tabular}
\end{table}

Before running the mining algorithms, nodes that cannot be refactored to a separate logic flow, such as \ActFlowStart and \ActFlowEnd nodes, are discarded.
Flows with refactor weight lower than the threshold $\beta$ ($5$ in our experiments) are also discarded.
Lastly, only benchmarks for which there exists a noticeable difference in results are considered in the evaluation, i.e., we ignore benchmarks for which the mining time and duplicated refactor weight across all pattern miner configurations does not vary by more than 1 second and $0.1\%$ respectively.
This results in a final collection of 693 benchmarks\footnote{The remaining 107 correspond to small code bases that were processed in less than 2 seconds by all configurations. The exact same amount of duplication was also found in each of the 107 benchmarks.}.
Table~\ref{tb:inst-stats} summarizes several statistics regarding the number of flows and nodes in these benchmarks.
Flows/nodes considered corresponds to the flows/nodes that are not discarded before mining.

\begin{table*}[t]
    \centering
    \caption{Maximum amounts of duplicated code found per benchmark set.}
    \label{tb:dup-stats}
    \begin{tabular}{ l | r r r r r r r r | r |}
        \textbf{Parameter} & \textbf{min} & $\mathbf{p = 0.25}$ & $\mathbf{p = 0.5}$ & $\mathbf{p = 0.75}$ & $\mathbf{p = 0.90}$ & $\mathbf{p = 0.95}$ & $\mathbf{p = 0.99}$ & \textbf{max} & \textbf{total} \\
        \hline
        Flows with duplicated code & $0.7\%$ & $8.8\%$ & $12.1\%$ & $15.8\%$ & $20.6\%$ & $23.4\%$ & $32.3\%$ & $42.9\%$ & - \\
        Duplicated nodes found & $0.7\%$ & $6.9\%$ & $9.9\%$ & $13.2\%$ & $18.1\%$ & $22.0\%$ & $30.9\%$ & $39.0\%$ & - \\
        Duplicated weight found & $14$ & $1185.00$ & $3623.00$ & $9956.00$ & $22632.00$ & $40896.40$ & $101675.72$ & $270756$ & $6817353$ \\
        \hline
    \end{tabular}
\end{table*}

Table~\ref{tb:dup-stats} shows some statistics regarding the amount of duplication found in the benchmark sets.
For each benchmark set and parameter, the maximum value obtained across all evaluated pattern mining configurations is considered.
The $p = x$ columns show the $x$-percentile values for each parameter.
For example, a value of $15.8\%$ in the $p = 0.75$ column of the 'Flows with duplicated code' row indicates that, in $75\%$ of the benchmarks, $15.8\%$ or less of the flows are found to contain duplicated code.
Note that these percentages consider the full universe of flows/nodes present in these benchmarks before pre-processing.

%Several statistics regarding the number of flows and nodes/edges in the benchmark sets are summarized in Table~\ref{tb:inst-stats}, as well as some statistics regarding the amount of duplication found in them.
%The $p = x$ columns show the $x$-percentile values for each parameter.
%For example, a value of 4955 flows in the $p = 0.75$ column indicates that $75\%$ of the benchmarks contain 4955 flows or less.
%We can see a considerable amount of duplication in these benchmarks: at least $9.5\%$ for half of the benchmarks, reaching as high as $38.8\%$.
%Note that these percentages consider the full universe of nodes present in these benchmark sets before pre-processing.
%If we consider the pre-processed benchmarks instead, the percentage of nodes that are duplicated reaches as high as $86.1\%$.

In order to solve the \ac{MCS} \ac{MaxSAT} instances, the PySAT~\cite{pysat} implementation of linear search~\cite{linear-search} is used.
Each \ac{MCS} extraction is run with a timeout of 10 seconds.
When the timeout is triggered, an approximate \ac{MCS} is retrieved from the best solution found by the linear search algorithm.
In order to prevent the pattern miner from becoming stuck due to occasional huge flows that result in hard \ac{MaxSAT} instances, the respective graphs are always removed from the active graph set whenever linear search fails to prove optimality, regardless of the refactor weight of the approximate \ac{MCS}.
In our experiments, we observed that timeouts are a rare occurrence: $0.1\%$ of a total of 267481 \ac{MCS} extractions for one of the configurations with graph pre-processing (Section~\ref{sec:spe-preprocess}) and inverted index (Section~\ref{sec:mining-inv-index}) enabled.
To solve the isomorphic pattern \ac{SAT} instances, we run PySAT with a timeout of 10 seconds as well.
In our experiments, PySAT was configured to use the Glucose \ac{SAT} solver (version 4.1)~\cite{glucose}.
All experiments were run on an AWS m5a.12xlarge instance with 128 GB of RAM.

%\todo{maximum patterns found per benchmark (with footnote on redundant patterns)}

\subsection{Algorithm Comparison}
\label{sec:eval-alg-comp}

\begin{table}[t]
    \centering
    \caption{Maximum and total mining time for different configurations of the pattern mining algorithms.}
    \label{tb:alg-eval}
    \begin{tabular}{ l | r r |}
        \textbf{Algorithm} & \textbf{max} & \textbf{total} \\
        \hline
        \Greedy & 2h09m12s & 1d09h52m39s \\
        \Lazy & 2h07m22s & 1d06h20m08s \\
        \DedupThenLazy & 1h13m32s & 19h49m41s \\
        \DedupThenLazyInvIndex & 48m39s & 6h30m45s \\
        \hline
    \end{tabular}
\end{table}

\begin{figure*}[t]
    \centering
    \begin{minipage}{0.31\textwidth}
        \centering
        \includegraphics[width=\linewidth]{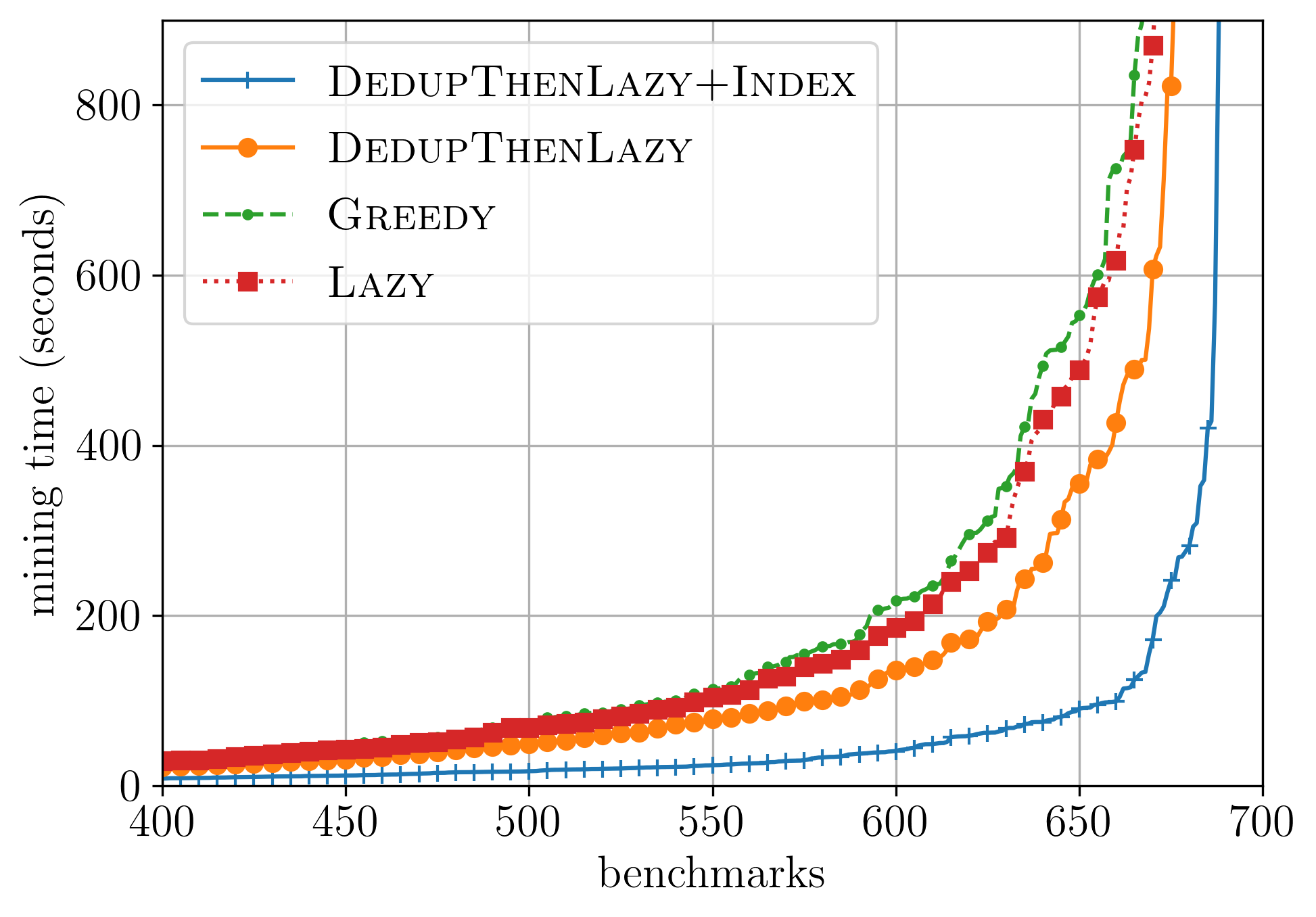}
        \caption{Distribution of mining time, in seconds, for different configurations of the pattern mining algorithms.}
        \label{fig:alg-eval}
    \end{minipage}%
    \hfill
    \begin{minipage}{0.31\textwidth}
        \centering
        \includegraphics[width=\linewidth]{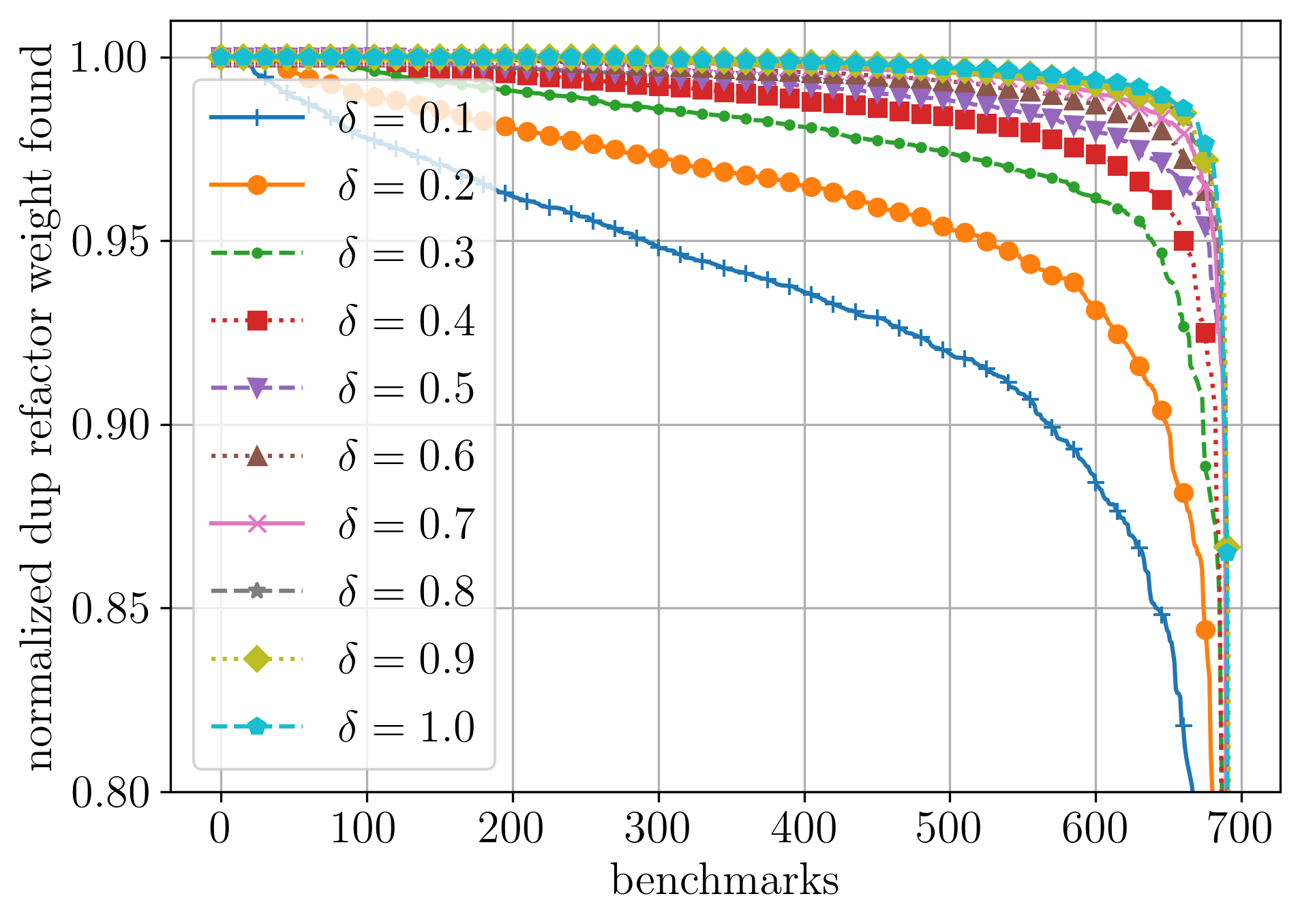}
        \caption{Normalized duplicated refactor weight found distribution for different values of the index's $\delta$ parameter.}
        \label{fig:inv-index-delta-refact-weight}
    \end{minipage}%
    \hfill
    \begin{minipage}{0.31\textwidth}
        \centering
        \includegraphics[width=\linewidth]{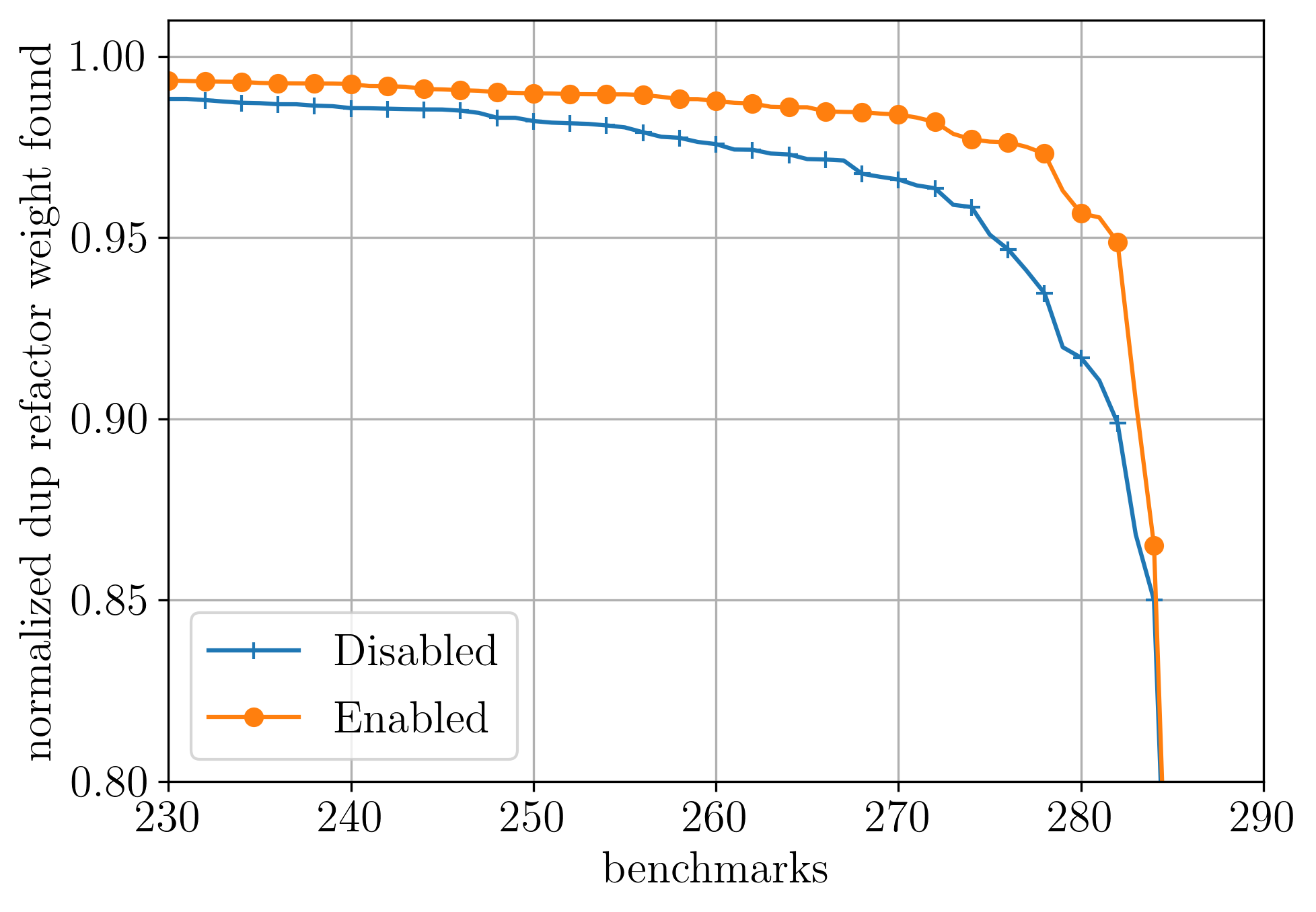}
        \caption{Normalized duplicated refactor weight found distribution with and without graph pre-processing.}
        \label{fig:inv-index-simp-refact-weight}
    \end{minipage}
\end{figure*}

Table~\ref{tb:alg-eval} compares the maximum and total mining time for the \Greedy (Section~\ref{sec:mining-greedy}) and \Lazy (Section~\ref{sec:mining-lazy}) algorithms.
Both algorithms were executed with graph pre-processing enabled.
\Lazy shows better performance than the original non-lazy version, being able to process all the benchmarks in $10.5\%$ less time.
Figure~\ref{fig:alg-eval} shows a distribution plot comparing the mining times of the algorithms.
For a given algorithm, each $(x, y)$ point in the plot indicates that, for $x$ benchmarks, the mining time of that algorithm is at most $y$.
For example, the $(600, 200)$ point in the line that corresponds to \Lazy indicates that 600 of the benchmarks are processed in 200 seconds or less by that algorithm.
Overall, we can see a small but noticeable reduction in mining times for \Lazy compared to \Greedy.

The performance improvement in terms of mining time was expected, since \Lazy adds graph pairs to the queue on an as-needed basis, resulting in a lower overhead incurred by queue updates.
Recall that the main advantage of \Lazy is a much shorter time-to-first-pattern, since, unlike \Greedy, it does not suffer from the major initialization overhead incurred by the eager initialization of the queue.
We observed an average time-to-first-pattern of $1$ second for \Lazy versus $96$ seconds for \Greedy, and maximum values of $25$ seconds and over $1$ hour and $12$ minutes respectively.

\subsection{Impact of De-duplication}

Table~\ref{tb:alg-eval} and Figure~\ref{fig:alg-eval} show the impact, in terms of mining time, of applying de-duplication (Section~\ref{sec:mining-isomorphic}) before running the \Lazy algorithm.
Note that the time spent on de-duplication is accounted for in the reported mining times.
Overall, we can see that the performance boost is quite significant.
In particular, \DedupThenLazy achieves a reduction of $34.6\%$ in total mining time compared to \Lazy.
Additionally, the reduction for the hardest benchmark is $42.3\%$.
This reduction makes sense because, as we observed in our experiments, \DedupThenLazy spends, on average, $13.8\%$ of mining time on de-duplication and removes about $19.4\%$ of the flows before running the pattern mining algorithm.

\subsection{Impact of Inverted Index}
\label{sec:eval-inv-index}

Table~\ref{tb:alg-eval} and Figure~\ref{fig:alg-eval} also compare the performance of \DedupThenLazy with and without the inverted index.
For this experiment, the $\delta$ parameter of the inverted index was set to $1.0$.
We can see that, compared to lazyfication and de-duplication, the inverted index has, by far, the largest overall positive impact in the performance of the pattern mining algorithm, achieving a reduction of $67.2\%$ in total mining time.
The performance improvement observed for the hardest instance is not as significant: a reduction of $33.8\%$.
This performance boost was expected since, with the inverted index, the algorithm only needs to compare each graph with a much smaller subset of graphs with overlapping combined edge labels, versus comparing all possible pairs.
Overall, the combination of lazyfication, de-duplication and the inverted index results in a total mining time reduction of $80.8\%$ compared to the original \Greedy algorithm.

Additional experiments were performed in order to evaluate the impact of the $\delta$ parameter on the performance of the mining algorithm.
We tested values of $\delta$ ranging from $0.1$ to $1.0$ in increments of $0.1$.
We observed that, in the best case, a value of $\delta = 0.1$ resulted in a small total mining time reduction of $4.5\%$ compared to $\delta = 1.0$.
Detailed results regarding mining time for the different values of $\delta$ are not shown due to space limitations.

Figure~\ref{fig:inv-index-delta-refact-weight} shows a distribution plot of duplicated refactor weight found for the different values of $\delta$.
These values are normalized against the largest duplicated refactor weight values found for each benchmark.
For a given value of $\delta$, each $(x, y)$ point in the plot indicates that, for $x$ benchmarks, the duplicated refactor weight found with $\delta$ is at least a fraction $y$ of the best value.
We can see that decreasing $\delta$ can have a significant negative impact on the amount of duplication that the algorithm is able to detect, particularly with $\delta \le 0.3$.
The impact is much less significant for $\delta \ge 0.6$.
However, using $\delta = 0.6$ results in a total mining time reduction of just $1.1\%$.
Overall, such a small mining time reduction does not compensate the negative impact on the algorithm's detection capabilities.

\subsection{Impact of Graph Pre-processing}
\label{sec:eval-graph-preproc}

\begin{table}[t]
    \centering
    \caption{Performance comparison of pattern mining with and without graph pre-processing.}
    \label{tb:pre-proc-eval}
    \begin{tabular}{ l | r r r | r |}
        & \multicolumn{3}{c |}{\textbf{\ac{MaxSAT} instances}} & \multicolumn{1}{c |}{\textbf{Mining}} \\
        \textbf{Pre-proc} & \textbf{total} & \textbf{optimal} & \textbf{total time} & \textbf{total time} \\
        \hline
        Disabled & 266584 & 266176 & 3h16m03s & 7h22m38s \\
        Enabled & 267481 & 267249 & 2h03m36s & 6h30m45s \\
        \hline
    \end{tabular}
\end{table}

Table~\ref{tb:pre-proc-eval} compares the performance of the lazy algorithm, with and without graph pre-processing, in terms of time spent solving \ac{MaxSAT} instances in addition to the total mining time.
Note that the time spent building the encoding and on pre-processing is accounted for in the reported times.
In both scenarios, the algorithm was executed with the inverted index enabled.
Enabling pre-processing results in the generation of 897 extra \ac{MaxSAT} instances.
This increase makes sense because pre-processing leads to the generation of smaller, and thus easier \ac{MaxSAT} instances.
Consequently, linear search is able to find larger \acp{MCS} before the timeout, resulting in more \acp{MCS} being generated before triggering the $\beta$ threshold of the pattern miner.
Note that, with pre-processing, optimality of the \ac{MCS} is proven for 1073 additional instances, which is more than the 897 extra ones.
Moreover, despite these extra instances, $37\%$ less time is spent in total solving \ac{MaxSAT} instances.
Overall, this translates to a reduction in total mining time of $11.7\%$.

Figure~\ref{fig:inv-index-simp-refact-weight} shows a distribution plot comparing the duplicated refactor weight found with and without graph pre-processing.
In order to improve readability, only benchmarks for which there was a variation of at least $0.1\%$ are considered.
We can see that a moderate improvement is achieved by enabling graph pre-processing.
This is expected since, as mentioned previously, more and larger \acp{MCS} are generated when pre-processing is enabled.

\section{Related Work}
\label{sec:related-work}

Many duplicated code detection techniques have been proposed in the literature for text-based programming languages.
Rattan et al.~\cite{clone-survey} wrote an extensive survey on this topic, where they classify these techniques into five main categories: text-based~\cite{text-parameterized-dup,text-near-dup,nicad,nicad-plus,text-substring-matching,text-lang-independent}, token-based~\cite{sourcerercc,ccfinder,cp-miner,token-incr-clone-detect,ccaligner,siamese}, tree-based~\cite{deckard,ast-clone-detect,clone-detection}, graph-based~\cite{graph-pdg-similar,gplag,ccsharp,graph-slicing,graph-scalable-accurate,ccgraph} and metrics-based~\cite{metrics-java-systems,metrics-reengineering,metrics-function-clones}.
Recently, several detectors based on machine learning have also emerged~\cite{oreo,cdlh,astnn,deepsim,ml-code-fragments,ml-tree-based-convolution,scdetector}.
These techniques (except most graph-based) only support duplicated code detection at a pre-defined granularity, i.e., are able to report, for example, groups of methods as duplicated, but are unable to do so for relatively small but frequent duplicated code patterns contained within said methods.
For example, such techniques may miss duplicated patterns like the one from Figures~\ref{fig:act-flow-simple} and~\ref{fig:act-flow-foreach} since only $60\%$ of those flows' logic is duplicated.

Graph-based detectors analyse the \acp{PDG} of the code blocks in order to detect duplicated code.
Typically, these approaches also rely on searching for isomorphic sub-graphs~\cite{graph-pdg-similar,graph-slicing,gplag,ccsharp}.
Because such detectors consider the \ac{PDG} of the graph, these are able to detect semantic duplicates with many syntactic changes.
%\footnote{In earlier versions of our system, we used a graph representation more similar to \acp{PDG} that included data dependencies instead of just the syntactic structure of the logic flows. This representation also enabled our approach to find semantic duplicates with significant syntactic differences, but discussions with experts of the target \ac{VPL} led to the conclusion that such duplicates were hard to analyse and understand. For this reason, we focused on detecting syntactic duplicates, but note that our approach is agnostic to the graph representation and can be seamlessly applied on \acp{PDG} in order to detect syntactically dissimilar duplicated code.}
However, scalability is an issue due to the hardness of checking sub-graph isomorphism.
Some graph-based detectors mitigate this by using heuristics in order to avoid some of these checks~\cite{gplag,ccsharp}, applying some limited form of pre-processing to the \acp{PDG}~\cite{ccsharp} or using approximate graph matching~\cite{ccgraph}.
To the best of our knowledge, ours is the first graph-based approach that solves the scalability issue by means of an inverted index.

Some approaches exist in the literature for detecting duplicated code in Simulink models~\cite{conqat,escan,scanqat,simone,opmcd}.
\SIMONE~\cite{simone} applies the \NiCaD~\cite{nicad,nicad-plus} text-based detector on textual representations of the models, thus sacrificing visual structure.
\ConQAT~\cite{conqat} uses heuristics to mine large duplicated code patterns from promising pairs of graphs and then group these patterns into clusters.
Due to the heuristic nature of the algorithm, it does not ensure intra-cluster consistency of the graph structure of the patterns.
Additionally, it is not able to detect smaller more frequent duplicated code patterns contained within larger less frequent ones.
\eScan~\cite{escan} solves these issues by using a combination of frequent sub-graph mining and maximal clique covering instead.
However, it has been shown that this approach does not scale in practice~\cite{scanqat,model-clones-practice}.
\ScanQAT~\cite{scanqat} mitigates this issue by combining \ConQAT and \eScan, but the reported results show a modest improvement over the latter.

\ac{MCS} extraction is a well-known problem with several important applications besides duplicated code detection~\cite{malware-sig-synth,mcs-app-chemical,mcs-app-graph-clustering,mcs-app-graph-dbs}.
Classical approaches solve the \ac{MCS} problem via reduction to maximum clique~\cite{mcs-clique}.
Our approach is closely related to more recent work that translates the problem to a constraint satisfaction~\cite{mcs-csp-vismara,mcs-csp-mccreesh} or an integer linear programming problem~\cite{mcs-ip}.
An alternative solution, proposed by McCreesh et al.~\cite{mcs-bnb}, uses branch and bound to search for \acp{MCS}.
A later iteration of this approach exploits reinforcement learning in order to learn a more effective branching heuristic~\cite{mcs-bnb-rl}.

Frequent sub-graph mining is closely related to the duplicated code pattern mining problem addressed in this work.
The typical solution is to follow a top-bottom approach that starts with a set of very small high frequency candidate common sub-graphs and iteratively extends them with new nodes/edges until their frequency falls below a given threshold~\cite{mining-apriori,subdue,gspan,gaston,dmtl}.
By nature, this approach maximizes sub-graph size while maintaining a pre-specified minimum frequency.
We decided to implement our own custom mining algorithms for duplicated code because, if a timeout is triggered, it is preferable to return a set of large high-impact patterns than a set of high frequency patterns with very few nodes/edges.

\section{Limitations and Discussion}
\label{sec:limits-discuss}

\begin{figure}[t]
    \centering
    \includegraphics[width=0.47\textwidth]{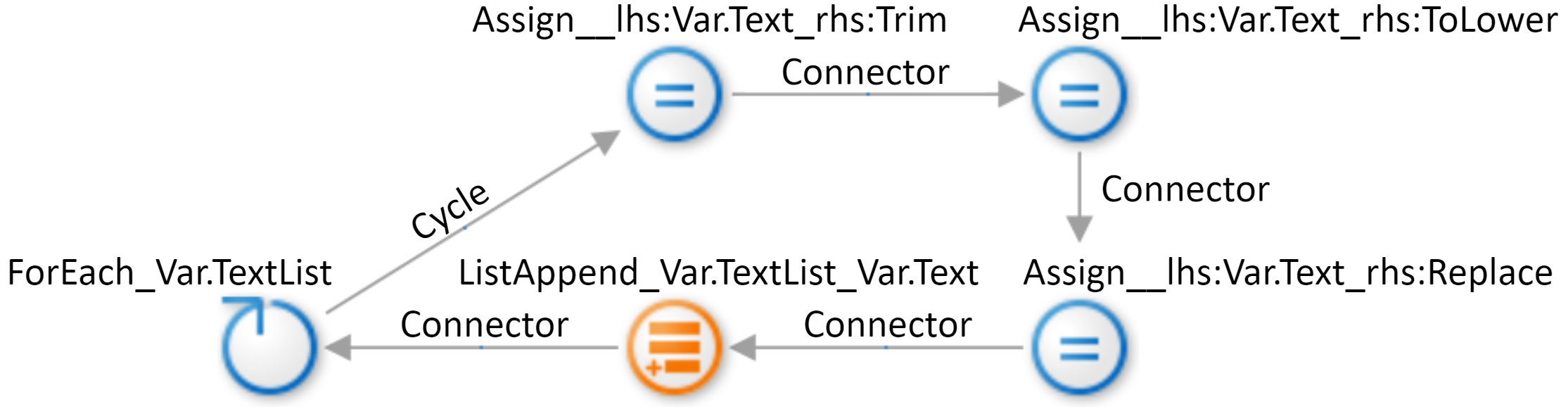}
    \caption{The labeled graph for the logic flow in Figure~\ref{fig:act-flow-foreach}.}
    \label{fig:labeled-act-flow}
\end{figure}

The precision and recall of the proposed approach strongly depends on the quality of the node and edge labels.
These were defined based on extensive iterative feedback from expert \TargetVPL developers.
The edge labels are set to their respective types in the logic flows, with the exception of \ActFlowSwitch~branches which consider the variable types and function calls that appear in the respective \ActFlowSwitch~conditions.
Node labels, however, can be quite sophisticated depending on their type.
A simple example is the \ActFlowIf~node label, which considers what kind of condition is being checked (e.g. null check) in addition to the respective variable types and function calls.
On the other hand, the label of an \ActFlowInstruction~node that performs a database access considers several characteristics, such as which tables are being accessed and which filters are being applied over which table columns.
Some normalizations were also performed, such as swapping the branches of \ActFlowIf~nodes if the condition is a negation of some Boolean expression.
Overall, the labels were tuned with the goal of maximizing the detection of type 3 duplicates that share the same graph structure.
Figure~\ref{fig:labeled-act-flow} shows the labeled version of the logic flow from Figure~\ref{fig:act-flow-foreach}.

Duplicated \ActFlowIf and \ActFlowSwitch nodes can only be refactored if at least one of their branches is also part of the duplicated code pattern.
However, the proposed mining algorithms do not capture this kind of constraint.
To circumvent this, some post-processing is applied to the patterns, immediately after extraction, in order to discard occurrences of such nodes.
Another option would be to sacrifice generality by adding additional clauses to the \ac{MaxSAT} encoding that enforce this contraint.

%Note that the algorithm allows some custom post-processing of the \acp{MCS} after their extraction (line~\ref{algl:mining-greedy-post-process}).
%This is supported for the following reasons:
%\begin{inlineenum}
%    \item It may be the case that $G_C$ contains some \ActFlowIf~or \ActFlowSwitch~node $v$ with none of its branches in the \ac{MCS}, i.e. no $(u', v') \in E_C$ exists such that $u' = v$.
%    Such nodes cannot be refactored to a separate logic flow, thus we discard them and the respective edges in post-processing.
%    \item Even though the refactoring weight of $G_C$ may satisfy the threshold $\beta$, it may be the case that some of its weakly connected components do not.
%    Such components are discarded as well during post-processing.
%\end{inlineenum}

As discussed in Section~\ref{sec:intro}, in earlier versions of our system, we used a graph representation more similar to \acp{PDG} that included data dependencies instead of just the syntactic structure of the logic flows.
This representation enabled the pattern miners to find semantic duplicates with significant syntactic differences, but discussions with \TargetVPL experts led to the conclusion that such duplicates were hard to analyse and understand.
For this reason, we focused on detecting type 3 duplicates with the same graph structure, but note that the proposed approach is agnostic to the graph representation and can be seamlessly applied on \acp{PDG} in order to detect syntactically dissimilar duplicated code.

Recall from Section~\ref{sec:intro} that the following requirements must be satisfied in order to provide a good user experience:
\begin{enumerate*}
    \item \label{ei:req-same-graph-struct} the graph structure of the duplicated code must be the same across its corresponding logic flows;
    \item \label{ei:req-mappings} the duplicated code detector must return the mappings of flow nodes to the duplicated code pattern in order for the tool to visually highlight the duplicated structure.
\end{enumerate*}
Most state-of-the-art detectors do not satisfy these requirements.
For example, \SIMONE~\cite{simone} applies text-based detection to Simulink models, thus losing the information needed for requirement~\ref{ei:req-mappings}.
The same applies to all non-graph based detectors for text-based languages, and even some of the graph-based like \CCGraph~\cite{ccgraph}, which performs approximate graph matching using graph kernels.
On the other hand, \ConQAT~\cite{conqat}, \eScan~\cite{escan}, \ScanQAT~\cite{scanqat} and \CCSharp~\cite{ccsharp} come close to satisfying these requirements.
However, we do not compare with these approaches for the reasons that follow.
\eScan's and \ScanQAT's source code is not publicly available.
\ConQAT's clustering step ignores the connections between nodes, thus not satisfying requirement~\ref{ei:req-same-graph-struct}.
Changing this requires replacing several list comparisons with isomorphism checks, which incurs a significant performance overhead.
Lastly, \CCSharp applies some filtering rules that are specific to \acp{PDG}.
Moreover, \CCSharp implements heuristics that prevent it from finding certain types of duplicated code patterns, such as duplicated sub-flows within large dissimilar flows or flows with dissimilar names.

\section{Conclusions and Future Work}
\label{sec:conclusion}

Duplicated code is an important form of technical debt that incurs a significant negative impact on software maintenance and evolution costs.
For this reason, for the past few decades, a large body of research has been dedicated to studying and addressing code duplication in text-based programming languages.
We propose a novel duplicated code detector for \TargetVPL that leverages the code's visual structure in order to provide helpful explanations of reported duplications.
Scalability is achieved by using an inverted index to avoid many unnecessary comparisons.
An extensive experimental evaluation carried on real-world \TargetVPL code bases show the effectiveness and scalability of the proposed solution.
This solution is currently deployed in the Architecture Dashboard\footnote{https://www.outsystems.com/platform/architecture-dashboard/}, a production static analysis tool for the \TargetVPL \ac{VPL}.

In the future, we plan to design and implement an incremental version of the pattern mining algorithm.
Incrementality has the potential to considerably reduce mining time, cutting down on computational resource costs and enabling real-time duplicated code detection.
Algorithms for mining duplicated code patterns that occur frequently within a single flow are being considered as well.
Lastly, we plan to exploit the tree structure of the patterns in order to provide a guided refactoring experience to the user, and eventually pursuit full automation of the refactoring process.

\begin{acks}
The authors would like to thank Alexandre Lemos, David Aparício and Ruben Martins for their valuable feedback and advice.
This work was supported by national funds through PT2020 with reference LISBOA-01-0247-FEDER-045309.
\end{acks}

\bibliographystyle{ACM-Reference-Format}
\balance
\bibliography{bibliography}

\end{document}